\documentclass[sigconf]{aamas} 

\usepackage{algorithm}
\usepackage[noend]{algpseudocode}
\usepackage{mathtools}
\usepackage{amsmath}
\usepackage{tikz}
\usetikzlibrary{calc}
\usepackage{todonotes}
\usepackage{dsfont}
\usepackage{enumitem}
\usepackage{tikz}
\usepackage{amsthm}
\usepackage{thm-restate}
\usepackage{booktabs}
\usepackage{multirow}
\usepackage{newfloat}
\usepackage{wrapfig}
\usepackage{dsfont}
\usepackage{comment}
\usepackage{multirow}
\usepackage{graphicx}
\usepackage{mleftright}

\title[Bayesian Persuasion Meets Mechanism Design: Going Beyond Intractability with Type Reporting]{Bayesian Persuasion Meets Mechanism Design:\\Going Beyond Intractability with Type Reporting}

\usepackage{balance} 

\acmConference[ArXiV]{ArXiV}{February}{2022} 
\copyrightyear{}
\acmDOI{}
\acmPrice{}
\acmISBN{}

\acmSubmissionID{596}

\newcommand{\NPHard}{$\mathsf{NP}$-hard}
\newcommand{\NP}{$\mathsf{NP}$}

\newcommand{\muvec}{\boldsymbol{\mu}}

\newcommand{\avec}{\textbf{a}}

\newcommand{\svec}{\mathbf{s}}
\newcommand{\sset}{\mathcal{S}}

\newcommand{\defeq}{\vcentcolon=}

\newcommand{\kvec}{{\boldsymbol{k}}}

\DeclareMathOperator*{\argmax}{arg\,max}

\newcommand{\rec}{\mathcal{R}}
\newcommand{\A}{\mathcal{A}}

\newcommand{\brset}{\mathcal{B}}

\newcommand{\nAct}{\ell}
\newcommand{\nType}{m}
\newcommand{\nState}{d}

\newcommand{\pvec}{{\boldsymbol{\xi}}}
\newcommand{\p}{\xi}
\newcommand{\pset}{\Xi}

\newcommand{\K}{\mathcal{K}}

\newtheorem{theorem}{Theorem}

\allowdisplaybreaks

\author{Matteo Castiglioni}
\affiliation{
  \institution{Politecnico di Milano}
  \city{Milan}
  \country{Italy}}
\email{matteo.castiglioni@polimi.it}
\author{Alberto Marchesi}
\affiliation{
	\institution{Politecnico di Milano}
	\city{Milan}
	\country{Italy}}
\email{alberto.marchesi@polimi.it}
\author{Nicola Gatti}
\affiliation{
	\institution{Politecnico di Milano}
	\city{Milan}
	\country{Italy}}
\email{nicola.gatti@polimi.it}

\begin{abstract}
	Bayesian persuasion studies how an informed \emph{sender} should partially disclose information so as to influence the behavior of self-interested \emph{receivers}.
	In the last years, a growing attention has been devoted to relaxing the assumption that the sender perfectly knows receiver's payoffs.
	The first crucial step towards such an achievement is to study settings where each receiver's payoffs depend on their unknown \emph{type}, which is randomly determined by a known finite-supported probability distribution.
	This begets considerable computational challenges, as computing a sender-optimal signaling scheme is inapproximable up to within any constant factor, even in basic settings with a single receiver.
	In this work, we circumvent this issue by leveraging ideas from mechanism design. 
	In particular, we introduce a \emph{type reporting} step in which the receiver is asked to report their type to the sender, after the latter has committed to a \emph{menu} defining a signaling scheme for each possible receiver's type.
	Surprisingly, we prove that, with a single receiver, the addition of this type reporting stage makes the sender's computational problem tractable.
	Then, we extend our Bayesian persuasion framework with type reporting to settings with multiple receivers, focusing on the widely-studied case of \emph{no inter-agent externalities} and binary actions.
	In such setting, we show that it is possible to find a sender-optimal solution in polynomial-time by means of the ellipsoid method, given access to a suitable polynomial-time separation oracle.
	This can be implemented for supermodular and anonymous sender's utility functions.
	As for the case of submodular sender's utility functions, we first approximately cast the sender's problem into a linearly-constrained mathematical program whose objective function is the multi-linear extension of the sender's utility.
	Then, we show how to find in polynomial-time an approximate solution to the program by means of a continuous greedy algorithm.
	This provides a $\left( 1 -\frac{1}{e}\right)$-approximation to the problem, which is tight.
\end{abstract}


\newcommand{\BibTeX}{\rm B\kern-.05em{\sc i\kern-.025em b}\kern-.08em\TeX}

\begin{document}




\maketitle 

\section{Introduction}

\emph{Bayesian persuasion}~\citep{kamenica2011bayesian} studies the problem faced by an informed agent (the \emph{sender}) trying to influence the behavior of other self-interested agents (the \emph{receivers}) via the partial disclosure of payoff-relevant information.
Agents' payoffs are determined by the actions played by the receivers and by an exogenous parameter represented as a \emph{state of nature}, which is drawn by a known prior probability distribution and observed by the sender only.
The sender commits to a public, randomized information-disclosure policy, which is customarily called \emph{signaling scheme}.
In particular, it defines how the sender should send private signals to the receivers, essentially deciding \emph{``who gets to know what"}.
These kinds of problems are ubiquitous in applications such as auctions and online advertising~\cite{bro2012send,emek2014signaling,badanidiyuru2018targeting,bacchiocchi2022public,castiglioni2022signaling}, voting~\citep{alonso2016persuading,cheng2015mixture,castiglioni2019persuading,castiglioni2020persuading,castiglioni2020public}, traffic routing~\citep{vasserman2015implementing,bhaskar2016hardness,castiglioni2020signaling}, recommendation systems~\cite{mansour2016bayesian}, security~\citep{rabinovich2015information,xu2016signaling}, and product marketing~\citep{babichenko2017algorithmic,candogan2019persuasion}.

In the classical Bayesian persuasion model by~\citet{kamenica2011bayesian}, the sender perfectly knows the payoffs of the receivers. This assumption is unreasonable in practice. Recently, some works tried to relax such an assumption. \citet{Castiglioni2020online} do that by framing the problem in an online learning framework, focusing on the single-receiver setting. They study the problem in which the sender repeatedly faces a receiver whose type during each iteration is unknown and selected beforehand by an adversary.
They design no-regret learning algorithms under full-information and partial-information feedback.
However, these algorithms require exponential running time, since even the offline problem in which the receiver's type is randomly selected according to a known finite-supported probability distribution is \NPHard\ to approximate up to within any constant factor.
\citet{castiglioniMulti2021} consider the problem with multiple receivers,
focusing on the classical model with binary actions and no-inter-agent-externalities~\citep{babichenko2017algorithmic,arieli2019private}, where each receiver's payoffs depend only on their action and the state of nature.
In the restricted setting in which each receiver has a \emph{constant} number of possible types, they show that the problem is intractable for supermodular and anonymous sender's utilities and design a no-$\left( 1-\frac{1}{e} \right)$-regret polynomial-time algorithm for submodular sender's utilities.
Let us remark that \citet{Castiglioni2020online} and \citet{castiglioniMulti2021} show that, in their respective settings, the design of polynomial-time no-regret algorithms is impossible due to the \NP-hardness of the underlining offline optimization problems in which the distribution over types is known.
Hence, the design of \emph{efficient} algorithms for the offline problem is the bottleneck to the design of efficient online learning algorithms.

In this work, we show how to circumvent this problem by leveraging ideas from mechanism design.
For the single-receiver setting, we introduce a \emph{type reporting} step in which the receiver is asked to report their type to the sender, after the latter has committed to a \emph{menu} defining a signaling scheme for each possible receiver's type.
Moreover, we extend the framework to accommodate multiple receivers with binary actions.
In such setting, we take advantage of the \emph{no-inter-agent-externalities} assumption to design a type reporting step that is independent among the receivers.
In particular, we introduce a type reporting step for each receiver, in which they are asked to report their type to the sender, after the latter has committed to a menu defining a \emph{marginal} signaling scheme for each possible receiver's type.
Then, the sender commits to a signaling scheme that is consistent with all the marginal signaling schemes.
By the no-inter-agent-externalities assumption, receivers' payoffs do \emph{not} depend on such a signaling scheme, and, thus, each receiver's decision problem in the type reporting step is well defined.

\subsection{Original Contribution}

In this work, we show that the introduction of a type reporting step makes the sender computational problem tractable.

For the single-receiver case, our main result is to show the existence of an optimal menu of \emph{direct} and \emph{persuasive} signaling schemes.
In the classical model in which the sender perfectly knows the receiver's payoffs, a signaling scheme is direct if signals represent action recommendations, while it is persuasive if the receiver is incentivized to follow recommendations.
We extend this definition to menus of signaling schemes. In particular, a menu is direct if the signals used by all the signaling schemes are action recommendations, whereas it is persuasive if the receiver has an incentive to follow action recommendations when they reported their true type.
By using this result, an optimal menu of signaling schemes can be computed efficiently by a \emph{linear program} (LP) of polynomial size.

In the multi-receiver setting, we focus on classes of sender's utility functions that are commonly studied in the literature, namely \emph{supermodular}, \emph{submodular}, and \emph{anonymous} functions~\citep{dughmi2017algorithmic,babichenko2017algorithmic,arieli2019private,Xu2020tractability}.
As in the single-receiver case, we show that there always exists an optimal sender's strategy using a menu of direct and persuasive marginal signaling schemes for each receiver.
This allows us to show that an optimal sender's strategy can be computed by solving an LP with polynomially-many constraints and exponentially-many variables.
This is possible in polynomial time by means of the \emph{ellipsoid method}, given access to a suitable polynomial-time separation oracle.
Such an oracle can be implemented for supermodular and anonymous sender's utility functions.
In the submodular case, the problem cannot be approximated within any factor better that $ 1-\frac{1}{e} $, since our problem generalizes the one without types, which is \NPHard\ to approximate up to within any factor better than $ 1-\frac{1}{e} $~\citep{babichenko2017algorithmic}.
However, we provide a polynomial-time algorithm that provides a tight $ \left( 1-\frac{1}{e} \right)$-approximation.
To do so, we show how to build a linearly-constrained mathematical program whose objective is the \emph{multi-linear extension} of the sender’s utility, having optimal value arbitrary close to that of an optimal sender's strategy.
Moreover, we show that, from a solution to this program, we can recover in polynomial time a sender's strategy having in expectation almost the same utility as the optimal value of the program.
Finally, we show how to find in polynomial time an approximate solution
to the program by means of a \emph{continuous greedy} algorithm.
This provides a $ \left( 1-\frac{1}{e} \right)$-approximation to the problem, which is tight.

\subsection{Related Works}

Most of the computational works on Bayesian persuasion study models in which the sender knowns the receiver's utility function exactly.
\citet{dughmi2016algorithmic} initiate these studies with the single-receiver case, while \citet{arieli2019private} extend their work to multiple receivers without inter-agent externalities, with a  focus on private signaling.
In particular, they focus on settings with binary actions for the receivers and a binary space of states of nature. They provide a characterization of an optimal signaling scheme in the case of supermodular, anonymous submodular, and super-majority sender's utility functions.
\citet{arieli2019private} extend this latter work by providing tight $\left(1-\frac{1}{e} \right)$-approximate signaling schemes for monotone submodular sender's utilities and showing that an optimal private signaling scheme for anonymous utility functions can be found efficiently.
\citet{dughmi2017algorithmic} generalize the previous model to settings with an arbitrary number of states of nature.
There are also some works focusing on public signaling with no inter-agent externalities, see, \emph{e.g.}, \citep{dughmi2017algorithmic}~and~\citep{Xu2020tractability}.

A recent line of research relaxed the assumption that the sender perfectly knows the receivers' utilities. \citet{Castiglioni2020online} and \citet{castiglioniMulti2021} study online problems with a single receiver and multiple receivers, respectively.
\citet{babichenko2021regret} study a game with a single receiver and binary actions in which the sender does not know the receiver utility, focusing on the problem of designing a signaling scheme that performs well for any possible receiver's utility function.
\citet{zu2021Learning} relax the perfect knowledge assumption, assuming that the sender and the receiver do not know the prior distribution over the states of nature. They study the problem of computing a sequence of persuasive signaling schemes that achieve small regret with respect to an optimal signaling scheme with knowledge of the prior distribution.

Our problem is also related to \emph{automated mechanism design}~\citep{conitzer2002complexity,guo2010computationally,vorobeychik2006empirical}.
The closest to our work is \cite{conitzer2003automated}, which studies a mechanism design problem between a mechanism designer and an agent.
The agent has a finite number of types and both the agent and the mechanism designer have a utility function that depends on the agent's type and on an outcome that the designer chooses from a finite set.
Moreover, the mechanism designer can commit to a menu specifying an outcome for each reported type.
The mechanism designer knows the receiver's probability distribution over types and their goal is to design an incentive compatible menu in order to maximize their utility.
The authors show that it is \NPHard\ to design an optimal menu, while if the mechanism is allowed to use randomization the problem can be solved in polynomial time.


\section{Formal Model}\label{sec:preliminaries}

We formally introduce the Bayesian persuasion framework with type reporting that we study in the rest of this work.
In particular, in Subsection~\ref{sec:single}, we describe the model with a single receiver, while in Subsection~\ref{sec:multi} we extend it to multi-receiver settings.

\subsection{Model with a Single Receiver} \label{sec:single}

The receiver has a finite set $A \defeq \{a_i\}_{i=1}^\nAct$ of $\nAct$ available actions and a type chosen from a finite set $K \defeq \{k_i\}_{i=1}^{\nType}$ of $\nType$ possible types.
For each type $k\in K$, the receiver's payoff function is $u^k: A \times\Theta\to [0,1]$, where $\Theta\defeq\{\theta_i\}_{i=1}^\nState$ is a finite set of $d$ states of nature.
We denote by $u_\theta^k(a)\in [0,1]$ the payoff obtained by the receiver of type $k \in K$ when the state of nature is $\theta \in \Theta$ and they play action $a \in A$.
%
The sender's payoffs are described by the functions $u^\mathsf{s}_\theta: A\to [0,1]$ for $\theta \in \Theta$.
%
As it is customary in Bayesian persuasion, we assume that the state of nature is drawn from a common prior distribution $\muvec\in\textnormal{int}(\Delta_\Theta)$, which is explicitly known to both the sender and the receiver.\footnote{$\textnormal{int}(X)$ is the {\em interior} of set $X$ and $\Delta_X$ is the set of all probability distributions over $X$. Vectors are highlighted in bold. For any vector $\mathbf{x}$, the value of its $i$-th component is $x_i$.}
The sender commits to a {\em signaling scheme} $\phi$, which is a randomized mapping from states of nature to {\em signals} for the receiver.
Formally, $\phi:\Theta \to \Delta_S$, where $S$ is a set of available signals.
%
For convenience, we let $\phi_\theta $ be the probability distribution employed by the sender to draw signals when the state of nature is $\theta \in \Theta$ and we denote by $\phi_\theta(s)$ the probability of sending signal $s \in S$.
%
Moreover, we slightly abuse the notation and use $\phi$ to also denote the probability distribution over signals induced by the signaling scheme $\phi$ and the prior distribution $\muvec$.

In the classical Bayesian persuasion framework by~\citet{kamenica2011bayesian} (without type reporting), the interaction between the sender and the receiver goes on as follows:
(i) the sender commits to a signaling scheme $\phi$ and the receiver is informed about it;
(ii) the sender observes the realized state of nature $\theta \sim \muvec$;
(iii) the sender draws a signal $s \in S$ according to $\phi_\theta$ and communicates it to the receiver;
(iv) the receiver observes $s $ and rationally updates their prior belief over $\Theta$ according to the {\em Bayes} rule; (v) the receiver selects an action maximizing their expected utility.

In step~(iv), after observing a signal $s\in S$, the receiver infers a posterior belief $\pvec^s\in\Delta_\Theta$ over the states of nature such that the component of $\pvec^s$ corresponding to state $\theta \in \Theta$ is:\footnote{We omit the dependency of $\pvec^s$ from $\phi$ as the signaling scheme that is actually used to compute the posterior will be clear from context. Moreover, for the ease of presentation, when we use notation $\text{Pr} \left\{ \cdot \right\}$ we assume that the set $S$ is finite, so that $\text{Pr} \left\{ \cdot \right\}$ is well defined. The notation can be easily generalized to the case of infinite sets $S$.}
\begin{equation}\label{eq:posterior}
\p^s_\theta\defeq \frac{\mu_\theta \Pr_{s'\sim \phi_\theta} \left\{ s'=s \right\} }{\Pr_{s' \sim \phi} \left\{ s'=s \right\} }.
\end{equation}
For the ease of notation, we let $\pset\defeq \Delta_\Theta$ be the set of receiver's posterior beliefs over states of nature. 
After computing $\pvec^s$, the receiver plays an action maximizing their utility in $\pvec^s$.
As it is customary in the literature~\cite{Castiglioni2020online,castiglioniMulti2021}, we assume that the receiver breaks ties in favor of the sender.
In the following, letting $\brset^k_\pvec \coloneqq \argmax_{a \in A} \sum_{\theta\in \Theta} \p_\theta u_\theta^k(a)$ be the set of actions that maximize the expected utility of the receiver of type $k \in K$ in any posterior $\pvec \in \pset$, we denote by $b^k_\pvec \in \argmax_{a \in \brset^k_\pvec} \sum_\theta \p_\theta u_\theta^s(a)$ the action in $\brset^k_\pvec$ that is actually played by the receiver of type $k$ in posterior $\pvec$.
%
%

In our \emph{Bayesian persuasion framework with type reporting}, the sender asks the receiver to report their type before observing the realized state of nature.
This enables the sender to increase their expected utility. 
%
In particular, before the receiver reports their type, the sender proposes to the receiver a \emph{menu} $\Phi=\{\phi^k\}_{k \in K}$ of signaling schemes, committing to send signals according to the signaling scheme $\phi^k $ if the receiver reports their type to be $k \in K$.
In details, the interaction goes on as follows:
(i) the sender proposes a menu $\Phi=\{\phi^{k}\}_{k \in K}$ to the receiver;
(ii) the receiver reports a type $k \in K$ that maximizes their expected utility given the proposed menu;
(iii) the sender observes the realized state of nature $\theta \sim \muvec$;
(iv) the sender draws a signal $s \in S$ according to $\phi^k_\theta$ and communicates it to the receiver;
%
%
%
finally, the interaction terminates with steps (iv) and (v) of the classical setting described above.
%

Notice that, in step (ii), the receiver of type $k \in K$ can compute their expected utility for each signaling scheme $\phi^{k'}$ in the menu as
\[
	\sum_{\theta \in \Theta} \mu_\theta \mathbb{E}_{s\sim \phi^{k'}_\theta} \left[ u^k_\theta \left( b^k_{\pvec^s} \right) \right],
\]	
and, then, they can report a type $k' \in K$ whose corresponding signaling scheme $\phi^{k'}$ maximizes their expected utility.
%

%
We focus on menus of signaling schemes that are \emph{incentive compatible} (IC), \emph{i.e.}, in which the receiver of type $k$ is incentivized to report their true type, for any $k \in K$.\footnote{Notice that, by a revelation-principle-style argument (see the book by~\citet{shoham2008multiagent} for some examples of these kind of arguments), focusing on IC menus of signaling schemes is w.l.o.g. when looking for a sender-optimal menu.}
Formally, a menu $\Phi=\{\phi^{k}\}_{k \in K}$ is IC if, for every type $k \in K$, the following constraints are satisfied:
\begin{equation} \label{eq:IC}
\hspace{-2mm} \sum_{\theta \in \Theta} \hspace{-0.5mm} \mu_\theta \mathbb{E}_{s\sim \phi^k_\theta} \hspace{-1mm} \left[ u^k_\theta \left(b^k_{\pvec^s} \right) \hspace{-0.5mm} \right]  \ge \sum_{\theta \in \Theta} \hspace{-0.5mm}\mu_\theta \mathbb{E}_{s\sim \phi^{k'}_\theta} \hspace{-1mm} \left[ u^k_\theta \left( b^k_{\pvec^s} \right) \hspace{-0.5mm} \right] \quad \forall k' \neq k.
\end{equation}

We say that a signaling scheme is \emph{direct} if $S=A$, which means that signals correspond to action recommendations for the receiver.
Moreover, we say that a direct signaling scheme is \emph{persuasive} if the receiver has an incentive to follow the action recommendations that they receive as signals, when they report their true type.
%
%
It is easy to check that a menu $\Phi =\{\phi^{k}\}_{k \in K}$ of direct and persuasive signaling schemes is IC if
\begin{align} \label{eq:ICDirect}
 \hspace{-1mm} \sum_{a \in A} \sum_{\theta \in \Theta} \hspace{-0.5mm} \mu_\theta \phi^k_\theta(a) u^k_\theta(a) \ge \hspace{-1mm} \sum_{a \in A} \max_{a' \in A } \sum_{\theta \in \Theta} \hspace{-0.5mm} \mu_\theta \phi^{k'}_\theta(a) u^k_\theta(a') \hspace{1mm} \forall  k'\neq k.
\end{align}
%

In the rest of this work, we will use the well-known equivalence between signaling schemes and distributions over receiver's posteriors (see~\citep{Kamenica2019Bayesian} for further details).
In particular, a signaling scheme $\phi$ in equivalent to a probability distribution $\gamma \in \Delta_{\Xi}$ over posteriors such that $\mathbb{E}_{\pvec \sim \gamma}[\pvec]=\muvec$, so that the expected utility of the receiver of type $k \in K$ under the  signaling scheme can be written as $\mathbb{E}_{\pvec \sim \gamma} \left[ \sum_{\theta \in \Theta} \p_\theta u^k_\theta \left( b^k_{\pvec} \right) \right]$.
Finally, when the distribution $\gamma \in \Delta_{\Xi}$ has finite support, we denote by $\gamma_\pvec$ the probability of $\pvec \in \Xi$ in $\gamma$.

\subsection{Model with Multiple Receivers} \label{sec:multi}

In a multi-receiver setting, there is a finite set $\rec \defeq \{r_i\}_{i=1}^{n}$ of $n$ receivers, and each receiver $r \in \rec$ has a type chosen from a finite set $\K_r \defeq \{k_{r, i} \}_{i=1}^{m_r}$ of $m_r$ different types.
We introduce $\K \defeq \bigtimes_{r \in \rec} \K_r$ as the set of type profiles, which are tuples $\kvec \in \K$ defining a type $k_r \in \K_r$ for each receiver $r \in \rec$.
Each receiver $r \in \rec$ has two actions available, defined by $\A_r \defeq \{a_0,a_1\}$.
We let $\A \defeq \bigtimes_{r \in \rec} \A_r$ be the set of action profiles specifying an action for each receiver.
%
%
%
The payoff of a receiver depends on the action played by them, while it does \emph{not} depend on the actions played by the other receivers, since we assume that there are \emph{no inter-agent externalities}.
Formally, a receiver $r \in \rec$ of type $k \in \K_r$ has a payoff function $u^{r,k}: \A_r \times\Theta\to [0,1]$.
%
%
%
%
The sender's payoffs depend on the actions played by all the receivers, and they are defined by $u^\mathsf{s} : \A \times \Theta \to [0,1]$.
For the ease of presentation, for every state of nature $\theta \in \Theta$, we introduce the function $f_\theta: 2^\rec \to [0,1]$ such that $f_\theta(R)$ represents the sender's payoff when the state of nature is $\theta$ and all the receivers in $R \subseteq \rec$ play action $a_1$, while the others play $a_0$.
In the rest of this work, we assume that the sender's payoffs are \emph{monotone non-decreasing} in the set of receivers playing $a_1$.
Formally, for each state $\theta \in \Theta$, we let $f_\theta(R) \le f_\theta(R')$ for every $R \subseteq R' \subseteq \rec$, while $f_\theta(\varnothing) = 0$ for the ease of presentation.
As it is customary, we focus on three families of functions: \emph{submodular}, \emph{supermodular}, and \emph{anonymous}.
We say that $f_\theta$ is submodular, respectively supermodular, if for $R,R' \subseteq \rec $: $f_\theta(R\cap R')+f_\theta(R\cup R')\le f_\theta(R)+f_\theta(R')$, respectively $f_\theta(R\cap R')+f_\theta(R\cup R')\ge  f_\theta(R)+f_\theta(R')$.
The function $f_\theta$ is {anonymous} if $f_\theta(R) = f_\theta(R')$ for all $R,R' \subseteq \rec : |R| = |R'|$.

With multiple receivers, the sender must send a signal to each of them. 
%
%
In this work, we focus on \emph{private} signaling, where each receiver has their own signal that is privately communicated to them.
Formally, there is a set $\sset_r$ of possible signals for each receiver $r \in \rec$.
Then, $\phi : \Theta \to \Delta_\sset$ is a signaling scheme,
where $\sset \defeq \bigtimes_{r \in \rec} \sset_r$ is the set of signal profiles, which are tuples $\svec \in \sset$ defining a signal $s_r \in \sset_r$ for each receiver $r \in \rec$. 
%
%
%
We denote by $\phi_\theta$ the probability distribution over signal profiles corresponding to state $\theta \in \Theta$, while we let $\phi_\theta(\svec)$ be the probability of sending $\svec \in \sset$.
Given a signaling scheme $\phi$, we define the resulting \emph{marginal signaling scheme} for receiver $r \in \rec$ as $\phi^r: \Theta \to \sset_r$.
Formally, for every $s \in \sset_r$, it holds that $\phi^r_\theta(s)= \Pr_{\svec \sim \phi_\theta} \left\{  s_r=s \right\}$.
%
Notice that receiver $r$'s posterior beliefs and expected utilities only depend on the marginal signaling scheme $\phi^r$.


The interaction between the sender and the receivers goes on as follows:
(i) the sender proposes to each receiver $r \in \rec$ a menu of marginal signaling schemes $\Phi^r = \{\phi^{r,k}\}_{k \in \K_r}$;
(ii) each receiver $r \in \rec$ reports a type $k_r \in \K_r$ such that $\phi^{r,k_r}$ is the marginal signaling scheme maximizing their expected utility;
(iii) the sender commits to a signaling scheme $\phi$ whose resulting marginal signaling schemes $\phi^r$ are such that $\phi^r \coloneqq \phi^{r,k_r}$ for all $r \in \rec$;
%
(iv) the sender observes the realized state of nature $\theta \sim \muvec$ and draws a signal profile $\svec \sim \phi_\theta$;
(v) each receiver $r \in \rec$ observes their signal $s_r$, rationally updates their prior belief over $\Theta$ according to the {\em Bayes} rule, and selects an action maximizing their expected utility.
Notice that the sender only needs to propose marginal signaling schemes to the receivers (rather than general ones), since the expected utility of each receiver only depends on their marginal signaling scheme, and \emph{not} on the others.
Thus, the sender can delay the choice of the (general) signaling scheme after types have been reported.

As customary, we assume that the receivers break ties in favor of the sender.
Since functions $f_\theta$ are monotone, this amounts to play $a_1$ whenever indifferent between the two actions.
Moreover, we say that a signaling scheme is direct and persuasive if  $S=\A$ and the receivers are better off playing recommended actions.
 We denote with $R\subseteq \rec$ the direct signal profile in which it is recommended to play $a_1$ to all the receiver in $R$ and $a_0$ to all the receiver in $\rec \setminus R$.

Similarly to the single-receiver case, we restrict the attention to IC menu of marginal signaling schemes.
Thus, in a multi-receiver setting, a sender's strategy is composed by an IC menu of marginal signaling scheme $\Phi^r = \{\phi^{r,k}\}_{k \in \K_r}$ for each receiver $r \in \rec$, and a set of signaling schemes $\{\phi^\kvec\}_{\kvec \in \K}$ (one per type profile possibly reported by the receivers) such that the resulting marginal signaling schemes satisfy $\phi^{\kvec, r} = \phi^{r,k_r}$ for all $\kvec \in \K$ and $r \in \rec$.
%
%


\subsection{Sender's Computational Problems}\label{sec:problemDef}

We consider the computational problem in which, given the probability distribution over the receivers' types, the sender wants to maximize their expected utility.
In the single-receiver case, the receiver's type $k \in K$ is drawn from a {known} distribution $\lambda \in \Delta_K$.
We call MENU-SINGLE the problem of computing an IC menu of signaling schemes $\Phi = \{ \phi^k \}_{k \in K}$ that maximizes the sender's expected utility, given a probability distribution $\lambda \in \Delta_K$ as input.
In the multi-receiver case, the types profiles $\kvec \in \K$ are drawn from a {known} distribution $\lambda \in \Delta_{\bar \K}$, where $\bar \K\subseteq \K$ is a subset of possible types vectors, \emph{i.e.}, the support of $\lambda$.
%
We call MENU-MULTI the problem of computing a sender's strategy---made by an IC menu of marginal signaling schemes $\Phi^r = \{\phi^{r,k}\}_{k \in \K_r}$ for each receiver $r \in \rec$ and a set of signaling schemes $\{\phi^\kvec\}_{\kvec \in \K}$---that maximizes the sender's expected utility, given $\lambda \in \Delta_{\bar \K}$ as input.\footnote{A polynomial-time algorithm for MENU-MULTI must run in time polynomial in the size of the instance and in the size of the support of the distribution $\lambda$. Notice that, in general, the latter may be exponential in the number of receivers $n$.}

\section{Single-receiver Problem}\label{sec:results_single}

%
%

We show how to solve MENU-SINGLE in polynomial time.

By using the equivalence between signaling schemes and distributions over posteriors (see Section~\ref{sec:single}), it is easy to check that an optimal menu of signaling schemes can be computed by the following LP~\ref{lp:infinite} with an \emph{infinite} number of variables, namely $\gamma^k \in \Delta_{\Xi}$ for $k \in K$.
In LP~\ref{lp:infinite}, the objective is the sender's expected utility assuming the receiver reports their true type, the first set of constraints encodes IC conditions, while the last one ensures that the distributions over posteriors correctly represent signaling schemes.
%
\begin{align}\label{lp:infinite}
	\max_{ \gamma } & \,\, \sum_{k \in K} \lambda_k  \mathbb{E}_{\pvec\sim \gamma^{k}} \sum_{\theta \in \Theta} \p_\theta u^\mathsf{s}_\theta \left( b^k_\pvec \right) \quad \text{s.t.} \\
	& \hspace{-0.3cm}\mathbb{E}_{\pvec\sim \gamma^{k}} \left[ \sum_{\theta \in \Theta} \p_\theta u^{k}_\theta \left( b^{k}_\pvec \right) \right] \ge \mathbb{E}_{\pvec\sim \gamma^{k'}} \left[ \sum_{\theta \in \Theta} \p_\theta u^{k}_\theta \left( b^k_\pvec \right) \right]  \forall k\neq k' \in K  \nonumber \\
	& \hspace{-0.3cm} \mathbb{E}_{\pvec\sim \gamma^{k}} \left[ \p_\theta \right] =\mu_\theta  \hspace{4cm}\forall  \theta \in \Theta, \forall k \in K \nonumber \\
	& \hspace{-0.3cm} \gamma^k \in \Delta_{\Xi} \hspace{6cm} \forall k \in K. \nonumber
\end{align}

As a first step, we show that there always exists an optimal solution to LP~\ref{lp:infinite} in which the probability distributions $\gamma^k \in \Delta_{\Xi}$ have {finite} support.
This allows us to compute an optimal menu of signaling schemes by solving an LP with a \emph{finite} number of variables.
In the following, for every $k \in K$ and $a \in A$, let $\pset^{k,a} \coloneqq \left\{ \pvec \in \pset:a \in \brset^k_\pvec \right\}$ and $\hat \pset^{k,a} \coloneqq \left\{ \pvec \in \pset: a=b^k_\pvec \right\}$.
Moreover, for every $\avec \in \bigtimes_{k \in K} A$, let $\pset^\avec \coloneqq \bigcap_{k \in K} \pset^{k,a_k}$ and $\hat \pset^\avec \coloneqq \bigcap_{k \in K} \hat \pset^{k,a_k}$, where $a_k$ is the $k$-th component of $\avec$.
Finally, let $\pset^*$ be such that $\pset^* \coloneqq \bigcup_{\avec \in \bigtimes_{k \in K} A} V(\pset^\avec)$, where $V(\pset^\avec)$ denotes the set of vertices of the polytope $\pset^\avec$.
The following Lemma~\ref{lm:finite} shows that there always exists an optimal menu of signaling schemes that can be encoded as probability distributions over $\pset^*$.
Formally, the lemma is proved by showing that the following LP~\ref{lp:GenFinite} is equivalent to LP~\ref{lp:infinite}.
%
\begin{subequations}\label{lp:GenFinite}
 	\begin{align} 
 		\max_{{\gamma}} & \,\, \sum_{k \in K} \lambda_k \sum_{\pvec \in \pset^*} \gamma^k_{\pvec} \sum_{\theta \in \Theta} \p_\theta u^{\mathsf{s}}_\theta \left( b^k_\pvec \right) \quad \text{s.t.} \label{lp:GenFiniteObj} \\
 		& \hspace{-0.2cm} \sum_{\pvec \in \pset^*} \hspace{-0.1cm} \gamma^{k}_{\pvec}  \sum_{\theta \in \Theta} \hspace{-0.1cm} \p_\theta u_\theta^k\left( b^k_\pvec \right) \hspace{-0.05cm} \ge \hspace{-0.1cm} \sum_{\pvec \in \pset^*} \hspace{-0.1cm} \gamma^{k'}_{\pvec} \hspace{-0.1cm} \sum_{\theta \in \Theta} \hspace{-0.1cm} \p_\theta u_\theta^k \left( b^k_{\pvec} \right)  \forall k\neq k' \hspace{-0.05cm} \in \hspace{-0.05cm} K \label{lp:GenFinite1}\\
 		& \hspace{-0.2cm} \sum_{\pvec \in \pset^*} \gamma^{k}_{\pvec}  \p_\theta =\mu_\theta \hspace{3.1cm} \forall k \in K, \forall \theta \in \Theta \label{lp:GenFinite2} \\
 		& \hspace{-0.2cm} \sum_{\pvec \in \pset^*} \gamma^k_\pvec = 1\hspace{4.6cm} \forall k \in K.
 	\end{align}
\end{subequations}
Intuitively, the result is shown by noticing that, once fixed the receiver's best responses to $\avec \in \bigtimes_{k \in K} A$, the sums over $\Theta$ in the objective and the constraints of LP~\ref{lp:infinite} are linear in the posterior $\pvec$, which allows to apply Carathèodory theorem to replace each posterior with a probability distributions over the vertices of $\pset^\avec$.

\begin{restatable}{lemma}{lemmaone}\label{lm:finite}
	In single-receiver instances, there always exists a sender-optimal menu of signaling schemes that can be encoded as probability distributions over the finite set of posteriors $\pset^*$.
\end{restatable}
%

Next, we show that there always exists an optimal menu of direct and persuasive signaling schemes, and that it can be computed in polynomial time by solving a polynomially-sized LP obtained by further simplifying LP~\ref{lp:GenFinite} (Theorem~\ref{thm:single}).
%
Notice that, in a Bayesian persuasion problem without type reporting, an optimal signaling scheme must employ a signal for each action profile $\avec\in  \bigtimes_{k \in K} A$.
Since these profiles are exponentially many, an optimal direct and persuasive signaling scheme cannot be computed in polynomial time by linear programming.
Indeed, without typer reporting, the problem has been shown to be \NPHard~\citep{Castiglioni2020online}.

An intuition behind the proof of Theorem~\ref{thm:single} is provided in the following.
Fix type $k \in K$ and action $a \in A$.
Suppose that an optimal menu of signaling schemes employs ${\gamma}^k \in \Delta_{\pset^*}$ for the type $k$, and that ${\gamma}^k$ has in the support two posteriors $\pvec^1,\pvec^2 \in  \hat \pset^{k,a}$ with probabilities $\gamma^{k}_{\pvec^1}$ and $\gamma^k_{\pvec^2}$.
Consider a new signaling scheme that replaces the two posteriors $\pvec^1$ and $\pvec^2$  with their convex combination $\pvec^* \in \Delta_{\pset^*}$, so that
\[
	\p^*_\theta= \frac{ \gamma^k_{\pvec^1} \p^1_\theta  +\gamma^k_{\pvec^2} \p^2_\theta }{ \gamma^k_{\pvec^1} +\gamma^k_{\pvec^2}} \,\, \text{for every} \,\, \theta \in \Theta \,\, 	\text{and}
\,\,
	\gamma^k_{\pvec^*}=\gamma^k_{\pvec^1} +\gamma^k_{\pvec^2}.
\]
Both $\pvec^1$ and $\pvec^2$ induce the same best response of the receiver of type $k$, and Objective~\eqref{lp:GenFiniteObj} and Constraints~\eqref{lp:GenFinite2} are linear in $\pvec$.
Hence, replacing the two posteriors with their convex combination $\pvec^*$ preserves the value of the objective, while maintaining the constraints satisfied.
The same does \emph{not} hold for Constraints~\eqref{lp:GenFinite1}, which are linear in the posterior only if we fix the best responses of all the receiver's types.
For Constraints~\eqref{lp:GenFinite1}, if we consider an inequality in which $\gamma^k$ appears in the left hand side, the sum over $\Theta$ is linear in $\pvec$ and 
\[
	\gamma^k_{\pvec^1} \sum_{\theta \in \Theta} \p^1_\theta  u^{k}_\theta(a) + \gamma^k_{\pvec^2} \sum_{\theta \in \Theta} \p^2_\theta  u^{k}_\theta(a) = \gamma^k_{\pvec^*} \sum_{\theta \in \Theta} \p^*_\theta u^k_\theta(a).
\]
Instead, if $\gamma^k$ appears in the right hand side, by the convexity of the max operator it hods:
\begin{align*}
	\gamma^{k}_{\pvec^1} \max_{a' \in A}& \sum_{\theta \in \Theta} \p^1_\theta u_\theta^{k'}(a') + \gamma^{k}_{\pvec^2} \max_{a' \in A} \sum_{\theta \in \Theta} \p^1_\theta u_\theta^{k'}(a') \\
	&\ge \max_{a' \in A} \left[ \gamma^{k}_{\pvec^1} \sum_{\theta \in \Theta} \p^1_\theta u_\theta^{k'}(a') + \gamma^{k}_{\pvec^2}  \sum_{\theta \in \Theta} \p^1_\theta u_\theta^{k'}(a') \right] \\
	& = \max_{a' \in A} \sum_{\theta \in \Theta} \p^*_\theta u^{k'}_\theta(a').
\end{align*}
Therefore, if we replace two posteriors that induce the same receiver's best responses with their convex combination, the left hand side of Constraints~\eqref{lp:GenFinite1} is preserved, while the value of the right hand side can only decrease, guaranteeing that Constraints~\eqref{lp:GenFinite1} remain satisfied.
By using this idea, we can join all the posteriors that induce the same best responses.
Finally, by resorting to the equivalence between signaling schemes and distributions over, we obtain the following LP~\ref{lp:direct} of polynomial size.
Hence, an optimal menu of signaling schemes can be computed in polynomial time.
 \begin{subequations}\label{lp:direct}
 	\begin{align} 
 		\max_{\phi,l} & \,\, \sum_{k \in K} \lambda_k  \sum_{\theta \in \Theta} \mu_\theta  \sum_{a \in A} \phi^k_\theta(a) u^\mathsf{s}_\theta(a) \quad \text{s.t.} \\
 		& \hspace{-5mm} \sum_{a \in A} \sum_{\theta \in \Theta} \mu_\theta   \phi^k_\theta(a) u^k_\theta(a) \ge \sum_{a \in A} l^{k,k'}_a \hspace{15mm} \forall k \neq k' \in K \label{lp:direct1}\\
 		& \hspace{-5mm} l^{k,k'}_a \ge \sum_{\theta \in \Theta} \mu_\theta \phi^{k'}_\theta(a) u^k_\theta(a') \hspace{11mm} \forall k\neq k' \in K, \forall a, a' \in A \label{lp:direct2}\\
 		& \hspace{-5mm} \sum_{\theta \in \Theta} \hspace{-0.5mm} \mu_\theta \phi^{k}_\theta(a) u^k_\theta(a) \hspace{-0.5mm}  \ge \hspace{-1mm} \sum_{\theta \in \Theta} \hspace{-0.5mm} \mu_\theta \phi^{k'}_\theta(a) u^k_\theta(a') \hspace{1mm} \forall k \in K, \forall a,a' \in A \label{lp:direct3}\\
 		& \hspace{-5mm} \sum_{a \in A} \phi^k_\theta(a) = 1 \hspace{36mm} \forall k \in K, \forall \theta \in \Theta.
 	\end{align}
 \end{subequations}
 
 Notice that Constraints~\eqref{lp:direct1}~and~\eqref{lp:direct2} are equivalent to the IC constraints for direct and persuasive signaling schemes, which are those specified in Equation~\eqref{eq:ICDirect}, 
 where $\max_{a' \in A} \sum_{\theta \in \Theta} \mu_\theta   \phi^{k'}_\theta(a) u^k_\theta(a')$ is the best response of the receiver of type $k \in K$ to the direct signal $a$ for the receiver of type $k' \in K$d .
 Moreover, Constraints~\eqref{lp:direct3} force the signaling schemes to be persuasive.
 
 \begin{restatable}{theorem}{theoremSingle}\label{thm:single}
 	In single-receiver instances, there always exists an optimal menu of direct and persuasive signaling schemes. Moreover, it can be computed in polynomial time.
 \end{restatable}

\section{Multi-receiver Problem} \label{sec:results_multi}

In this section, we switch the attention to MENU-MULTI.
%
As we will show in the following (Theorem~\ref{thm:direct}), given any multi-receiver instance, there always exists an optimal sender's strategy that uses menus of direct and persuasive marginal signaling schemes.
This allows us to formulate the sender's problem as the following LP~\ref{lp:multi}, which will be crucial for the results in the rest of this section.

Since $\phi^k_\theta(a_0)=1-\phi^k_\theta(a_1)$ for every $r \in \rec$, $k \in \bar \K_r$, and $\theta \in \Theta$, by letting $x^{r,k}_\theta=\phi^k_\theta(a_1)$ we can formulate the following LP:
%
%
\begin{subequations}\label{lp:multi}
	\begin{align}
	\max_{\phi\ge0,x\ge0} & \,\, \sum_{\theta \in \Theta} \mu_\theta \sum_{\kvec \in \bar \K} \lambda_\kvec \sum_{R \subseteq \rec} \phi^{\kvec}_\theta(R) f_\theta(R) \quad \text{s.t.} \\
	& \hspace{-1cm} \sum_{R \subseteq \rec:r \in R} \phi^{\kvec}_\theta(R)  = x^{r,k_r}_\theta \hspace{1.4cm} \forall \kvec \in \bar \K, \forall r \in \rec, \forall \theta \in \Theta \label{lp:multi1}\\
	& \hspace{-1cm} \sum_{\theta \in \Theta} \mu_\theta x^{r,k}_\theta u^{r,k}_\theta(a_1) +\sum_{\theta \in \Theta} \mu_\theta \left( 1-x^{r,k}_\theta \right) u^{r,k}_\theta(a_0) \nonumber \\
	& \hspace{-0.5cm} \ge l^{r,k,k'}_{a_1}+l^{r,k,k'}_{a_0} \hspace{2cm} \forall r \in \rec, \forall  k\neq k' \in \K_r \label{lp:multi2}\\
	& \hspace{-1cm} l^{r,k,k'}_{a_1} \ge \sum_{\theta \in \Theta} \mu_\theta x^{r,k'}_\theta u^{r,k}_\theta(a) \nonumber\\
	& \hspace{2.55cm} \forall r \in \rec, \forall a \in \A_r, \forall k\neq k' \in \K_r \label{lp:multi3}\\
	& \hspace{-1cm} l^{r,k,k'}_{a_0} \ge \sum_{\theta \in \Theta} \mu_\theta \left( 1-x^{r,k'}_\theta \right) u^{r,k}_\theta(a) \nonumber \\
	& \hspace{2.55cm} \forall r \in \rec, \forall a \in \A_r, \forall k \neq k' \in \K_r \label{lp:multi4}\\
	& \hspace{-1cm} \sum_{\theta \in \Theta} \mu_\theta x^{r,k}_\theta \left[ u^{r,k}_\theta(a_1) -u^{r,k}_\theta(a_0) \right] \ge 0 \hspace{1.05cm} \forall r \in \rec, \forall k \in \K_r \label{lp:multi5}\\
	& \hspace{-1cm} \sum_{\theta \in \Theta} \mu_{\theta} \left( 1-x^{r,k}_\theta \right) \left[ u^{r,k}_\theta(a_0) -u^{r,k}_\theta(a_1) \right] \ge 0 \hspace{2mm} \forall r \in \rec, \forall k \in \K_r \label{lp:multi6}\\
	& \hspace{-1cm} \sum_{R \subseteq \rec}  \phi^{\kvec}_\theta(R) = 1 \hspace{3.4cm} \forall \kvec \in \bar \K, \forall \theta \in \Theta. \label{lp:multi7} 
	\end{align}
\end{subequations}
%
%
In the LP, Constraints~\eqref{lp:multi1} represent consistency conditions ensuring that the general signaling scheme $\phi^\kvec$ results in the marginal signaling schemes $\phi^{k,r}$, which are defined by means of variables $x_\theta^{r,k}$.
Constraints~\eqref{lp:multi2},~\eqref{lp:multi3},~and~\eqref{lp:multi4} represent IC constraints for the menus of marginal signaling schemes, where, as in LP~\ref{lp:direct}, we use Constraints~\eqref{lp:multi3}~and~\eqref{lp:multi4} with variables $l^{r,k,k'}_{a_0}$, $l^{r,k,k'}_{a_1}$ to compute receivers' expected utilities of playing a best response.
Finally, Constraints~\eqref{lp:multi5}~and~\eqref{lp:multi6} encode the persuasiveness conditions, while Constraints~\eqref{lp:multi7} require the signaling scheme be well defined.

Next, we prove our main existence result supporting LP~\ref{lp:multi}.

\begin{restatable}{theorem}{theoremDirectMulti} \label{thm:direct}
	In multi-receiver instances, there always exists an optimal sender's strategy that uses menus of direct and persuasive marginal signaling schemes.
\end{restatable}

\subsection{Supermodular/Anonymous Sender's Utility}\label{sec:results_multiSup}

LP~\ref{lp:multi} has an exponential number of variables and a polynomial number of constraints.
Nevertheless, we show that it is possible to apply the ellipsoid algorithm  to its dual formulation in polynomial time, provided access to a suitably-defined separation oracle.

\begin{restatable}{theorem}{theoremMulti}\label{thm:oracle}
	Given access to an oracle that solves $\max_{R \subseteq \rec} f_\theta(R)+\sum_{r\in R} w_r$ for any $\boldsymbol{w} \in \mathbb{R}^n$, there exists a polynomial-time algorithm that finds an optimal sender's strategy in any multi-receiver instance.
\end{restatable}

An oracle that solves $\max_{R \subseteq \rec} f_\theta(R)+\sum_{r\in R} w_r$ can be implemented in polynomial time for supermodular and anonymous functions, as shown by~\citet{dughmi2017algorithmic}.
%
As a consequence, we obtain the following corollary.

\begin{restatable}{corollary}{corollaryOracle}
	In multi-receiver instances with supermodular or anonymous sender's utility functions, there exists a polynomial-time algorithm that computes an optimal sender's strategy.
\end{restatable}

\subsection{Submodular Sender's Utility}\label{sec:results_multiSub}

In this section, we show how to obtain in polynomial time a $\left( 1-\frac{1}{e} \right)$-approximation to an optimal sender's strategy in instances with submodular utility functions, modulo an additive loss $\epsilon > 0$.
%
This is the best approximation result that can be achieved in polynomial time, since, as it follows from results in the literature, it is~\NPHard\ to obtain an approximation factor better than $ 1-\frac{1}{e} $.
Indeed, if we consider settings without types, \emph{i.e.}, in which $|\K_r|=1$ for all $r \in \rec$, the problem reduces to computing an optimal signaling scheme when the sender knows receivers' utilities.
Then, in the restricted case in which there are only two states of nature, \citet{babichenko2017algorithmic} show that, for each $\epsilon>0$, it is \NPHard\ to provide a $\left(1-\frac{1}{e}+\epsilon \right)$-approximation of an optimal signaling scheme.

Then, the following theorem provides a tight approximation algorithm that runs in polynomial time.

\begin{theorem} \label{thm:submodular}
	For each $\epsilon>0$, there exists an algorithm with running time polynomial in the instance size and $\frac{1}{\epsilon}$ that returns a sender's strategy with utility at least $\left( 1- \frac{1}{e} \right) OPT-\epsilon$ in expectation, where $OPT$ is the sender's expected utility in an optimal strategy.
\end{theorem}

In order to prove the result, we reduce the problem of computing the desired (approximate) sender's strategy to solving the following linearly-constrained mathematical program (Program~\ref{lp:multiSub}).
The program exploits the fact that, as we will show next, there always exists an ``almost'' optimal sender's strategy in which the sender employs signaling schemes $\phi^{\kvec}$ (for $\kvec \in \bar \K$) such that the distributions $\phi^{\kvec}_\theta$ are $q$-uniform over the set $2^{\rec}$. 
In particular, we say that a distribution is $q$-uniform if it follows a uniform distribution on a multiset of size $q$, where we denote by $[q]$ the set $\{1,\dots,q\}$.
%
%
Then, the mathematical program reads as follows.
\begin{subequations} \label{lp:multiSub}
\begin{align}
\max_{x} & \,\,  \sum_{\theta \in \Theta} \mu_\theta \sum_{\kvec \in \bar \K} \lambda_\kvec \frac{1}{q} \sum_{j \in [q]} F_\theta \left( x^{j,\kvec,\theta} \right) \quad \text{s.t.}\\
&\hspace{-0.3cm}\sum_{j \in [q]} \frac{1}{q} x^{j,\kvec,\theta}_{r}\le x^{r,k_r}_\theta \hspace{1.45cm} \forall r \in\rec, \forall \kvec \in \bar \K,\forall \theta \in \Theta \label{lp:multiRelaxed}\\ 
&\hspace{-0.3cm}\sum_{\theta \in \Theta} \mu_\theta x^{r,k}_\theta u^{r,k}_\theta(a_1) +\sum_{\theta\in \Theta} \mu_\theta \left( 1-x^{r,k}_\theta \right) u^{r,k}_\theta(a_0) \nonumber\\
&\hspace{-0.3cm}\hspace{0.5cm}\ge l^{k,k'}_{a_1}+l^{k,k'}_{a_0} \hspace{2.2cm} \forall r \in \rec, \forall k\neq k' \in \K_r \label{lp:multiSub1}\\
&\hspace{-0.3cm}l^{k,k'}_{a_1} \ge \sum_{\theta\in \Theta} \mu_\theta x^{r,k'}_\theta u^{r,k}_\theta(a) \nonumber \\
&\hspace{-0.3cm}\hspace{3.35cm}\forall r \in \rec, \forall a \in \A_r, \forall k \neq k' \in \K_r \label{lp:multiSub2} \\
&\hspace{-0.3cm}l^{k,k'}_{a_0} \ge \sum_{\theta\in \Theta} \mu_\theta \left( 1-x^{r,k'}_\theta \right) u^{r,k}_\theta(a) \nonumber\\
&\hspace{-0.3cm}\hspace{3.35cm}\forall r \in \rec, \forall a \in \A_r , \forall k\neq k' \in \K_r \label{lp:multiSub3}\\
&\hspace{-0.3cm}\sum_{\theta\in \Theta} \mu_\theta x^{r,k}_\theta \left[ u^{r,k}_\theta(a_1) -u^{r,k}_\theta(a_0) \right] \ge 0 \hspace{0.65cm}\forall r \in \rec,k \in \K_r \label{lp:multiSub4}\\
&\hspace{-0.3cm}\sum_{\theta\in \Theta} \hspace{-0.5mm} \mu_\theta \hspace{-0.5mm} \left( \hspace{-0.5mm} 1 \hspace{-0.5mm} -x^{r,k}_\theta \right) \hspace{-0.1cm}\left[ u^{r,k}_\theta(a_0) \hspace{-0.5mm} -u^{r,k}_\theta(a_1) \right] \hspace{-0.5mm} \ge 0 \hspace{0.01cm}\forall r \in \rec, \forall k \in \K_r \label{lp:multiSub5}\\
&\hspace{-0.3cm}0 \le x^{r,k}_\theta \le 1 \hspace{2.5cm} \forall r \in \rec, \forall k \in \K_r, \forall \theta \in \Theta \label{lp:multiSub6}\\
&\hspace{-0.3cm}0 \le x^{j,\kvec,\theta}_r \le 1 \hspace{1.1cm} \forall j \in [q], \forall r \in \rec, \forall k \in \K_r, \forall \theta \in \Theta. \label{lp:multiSub7}
\end{align}
\end{subequations}
In Program~\ref{lp:multiSub}, each variable $x^{r,k}_\theta$ represents the probability $\phi^k_\theta(a_1)$ that the sender recommends action $a_1$ to receiver $r \in \rec$ of type $k \in \K_r$ in state $\theta \in \Theta$.
Constraints~\eqref{lp:multiSub1}--\eqref{lp:multiSub6} force the marginal signaling schemes to be well defined, where Constraints~\eqref{lp:multiSub1},~\eqref{lp:multiSub2},~and~\eqref{lp:multiSub3} encode the IC conditions, Constraints~\eqref{lp:multiSub4}~and~\eqref{lp:multiSub5} ensure the persuasiveness property, and Constraints~\eqref{lp:multiSub6} require the marginal signaling schemes to be feasible, \emph{i.e.}, $\phi^{r,k}_\theta(a_1)+\phi^{r,k}_\theta(a_0)=1$ and $\phi^{r,k}_\theta(a)\ge0$ for every $a \in \{a_0,a_1\}$.
Moreover, the program uses variables $x^{j,\kvec,\theta}_r \in \{0,1\}$ to represent whether the recommended action to receiver $r \in \rec$ is $a_1$ or $a_0$ in the $j$-th action profile in the support of $\phi^{\kvec}_\theta$.
Notice that we relaxed these variables to $x^{j,\kvec,\theta}_{r}\in[0,1]$ and use the multi-linear extension of the sender's utility functions $f_\theta$, which, for every $\theta \in \Theta$, reads as
\[
	F_\theta(x) \coloneqq \sum_{R \subseteq \rec} f_\theta(R)  \prod_{r \in R}{x_r} \prod_{r \notin R}(1- x_r).
\]
%
Moreover, we also relax the constraints ensuring the consistency of the marginal signaling schemes, namely Constraints~\eqref{lp:multiRelaxed}, by replacing the condition $\sum_{j \in [q]} \frac{1}{q} x^{j,\kvec,\theta}_{r}= x^{r,k_r}_\theta $ for all $ r \in \rec,\kvec \in \bar \K,\theta \in \Theta$ with $\sum_{j \in [q]} \frac{1}{q} x^{j,\kvec,\theta}_{r}\le x^{r,k_r}_\theta $ for all $ r \in \rec,\kvec \in  \bar \K,\theta \in \Theta$.

In order to reduce the problem of computing the desired sender's strategy to solving Program~\ref{lp:multiSub}, we need the following two lemmas (Lemma~\ref{lm:subUniform} and Lemma~\ref{lm:subDirect}).
We show that the value of Program~\ref{lp:multiSub} for a suitably-defined $q$ approximates the value of an optimal sender's strategy (\emph{i.e.}, an optimal solution to LP~\ref{lp:multi}) and that, given a solution to Program~\eqref{lp:multiSub}, we can recover a sender's strategy with approximately the same expected utility for the sender.
Our result is related to those in~\citep{dughmi2017algorithmic} for the case without types.
However, \citet{dughmi2017algorithmic} use a probabilistic method to show the existence of an ``almost'' optimal signaling scheme that uses $q$-uniform distributions over the signals.
This approach cannot be applied to our problem since it slightly modifies the receivers' utilities.
In the case of persuasiveness constraints, they show how to maintain feasibility.
However, this approach does \emph{not} work for the IC constraints.
We propose a different technique based on the fact that LP~\ref{lp:multi} has a polynomial number of constraints.
Let $\beta$ be the number of constraint of LP~\ref{lp:multiSub}.
Notice that $\beta$ is polynomial in the size of the LP.
We show that, for each $\epsilon>0$, there exist a $q$ such that LP~\ref{lp:multiSub} has value at least $OPT-\epsilon$, where $OPT$ is the value of an optimal sender' strategy.

\begin{restatable}{lemma}{lemmaSubUniform}\label{lm:subUniform}
	For each $\epsilon>0$, the optimal value of Program~\ref{lp:multiSub} with $q=\left\lceil \frac{\beta}{\epsilon} \right\rceil$ is at least $OPT-\epsilon$, where $OPT$ is the value of an optimal sender's strategy and $\beta$ is the number of constraints of LP~\ref{lp:multi}.
\end{restatable}

Then, we show how to obtain a signaling scheme given a solution of Program~\ref{lp:multiSub}.
\citet{dughmi2017algorithmic} build a signaling scheme by using a technique whose generalization to our setting works as follows.
Given a state of nature $\theta \in \Theta$ and a vector of types $\kvec \in \bar \K$, it selects a $j \in [q]$ uniformly at random and recommends signal $a_1$ to receiver $r \in \rec$ with probability $x^{j,\kvec,\theta}_r$, while it recommends $a_0$ otherwise.
By definition of multi-linear extension, using this technique the sender achieves expected utility equal to the value of the given solution to Program~\ref{lp:multiSub}.
However, this signaling scheme uses an exponential number of signal profiles, and, thus, it cannot be represented explicitly.
In the following lemma, we show how to obtain a sender's strategy in which signaling schemes use a polynomial number of signal profiles.

\begin{restatable}{lemma}{lemmaSubDirect} \label{lm:subDirect}
	Given a solution to Program~\ref{lp:multiSub} with value $APX$, for each $\iota>0$, there exists an algorithm with running time polynomial in the instance size and $\iota$ that returns a sender's strategy with utility at least $APX-\frac{n}{q}-\iota$ in expectation. Moreover, such sender's strategy employs signaling schemes using polynomially-many signal profiles.
\end{restatable}

\begin{proof}
	%
	Let $x$ be a solution to Program~\ref{lp:multiSub} with value $APX$.
	Next, we show how to obtain the desired sender's strategy.

	First, we build a new solution to Program~\ref{lp:multiSub} such that Constraints~\eqref{lp:multiRelaxed} are satisfied with equality.
	Since the functions $F_\theta$ are monotonic, we can simply obtain such solution 
	%
	by increasing the values of variables $x^{j,\kvec,\theta}_r$.
	It is easy to see that, by the monotonicity of $F_\theta$, the objective function does \emph{not} decrease.

	Then, we obtain an ``almost binary'' solution by applying, for every $\theta \in \Theta$ and $\kvec \in \bar \K$, the procedure outlined in Algorithm~\ref{alg:lemma3}.
	An a first operation, the algorithm iterates over the receivers, doing the operations described in the following for each receiver $r \in \rec$.
	
	For each $j \in [q]$, the algorithm computes an estimate of the following partial derivative  
	\begin{align*}
		&\frac{\partial F_\theta\big( x^{j,\kvec,\theta} \big)}{\partial x^{j,\kvec,\theta}_r} \hspace{-0.1cm} = \\
		&\hspace{0.5cm} \sum_{R\subseteq \rec\setminus \{r\}} \hspace{-0.2cm} \big[ f_\theta(R\cup\{r\})-f_\theta(R) \big] \prod_{r' \in R}  {x^{j,\kvec,\theta}_r} \hspace{-2mm}\prod_{r' \notin R, r' \neq r} \left( 1- x^{j,\kvec,\theta}_r \right), 
	\end{align*}
	This is accomplished by drawing $\sigma=\frac{-8}{\iota^2} n^2 \log \frac{p}{2}$ samples of the random variable $f_\theta(\tilde R \cup \{r\})-f_\theta(\tilde R)$ (with $p=\frac{\iota}{2|\bar \K| d q n}$), where $\tilde R \subseteq \rec$ is obtained by randomly picking each receiver $r'\in \rec : r'\neq r$ independently with probability $x^{j,\kvec,\theta}_{r'}$.
	%
	It is easy to see that the expected value of the random variable is exactly equal to value of the partial derivative above.
	Letting $\tilde e^{j,\kvec,\theta}_{r}$ be the empirical mean of the samples,
	%
	by an Hoeffding bound, we get
	\[
		\Pr \left\{ \, \left| \tilde e^{j,\kvec,\theta}_r- \frac{\partial F_\theta \big( x^{j,\kvec,\theta} \big)}{\partial x^{j,\kvec,\theta}_r} \right| \ge \frac{\iota}{4n} \right\} \le p.
	\]
	Moreover, consider the event $\mathcal{E}$ in which $\Big| \tilde e^{j,\kvec,\theta}_r- \frac{\partial F_\theta \big( x^{j,\kvec,\theta} \big)}{\partial x^{j,\kvec,\theta}_r} \Big| \le \frac{\iota}{4n}$ for all $ j \in [q],\kvec \in \bar \K,\theta \in \Theta,$ and $r \in \rec$.
	By a union bound, the event $\mathcal{E}$ holds with probability at least $1 - p |\bar \K| d q n$.

	%
	As a second step, the algorithm re-labels the indexes so that, if $j < j' $, then $\tilde e^{j,\kvec,\theta}_{r} \geq \tilde e^{j',\kvec,\theta}_{r}$.
	%
	Notice that the value of the partial derivative with respect to $x^{j,\kvec,\theta}_r$ does \emph{not} depend on its value.
	Hence, given two indexes $j<j'$, by ``moving'' a value $t$ from $x^{j',\kvec,\theta}_r$ to $x^{j,\kvec,\theta}_r$, the sum $\sum_{j \in [q]} F_\theta \big( x^{j,\kvec,\theta}_r \big)$ decreases at most of 
	\begin{align*}
    &	t \left(  \frac{\partial F_\theta(x^{j',\kvec,\theta})}{\partial x^{j',\kvec,\theta}_r}-\frac{\partial F_\theta(x^{j,\kvec,\theta})}{\partial x^{j,\kvec,\theta}_r} \right) \le t \left( \frac{\iota}{2n} +\tilde{e}^{j',\kvec,\theta}_r-\tilde{e}^{j,\kvec,\theta}_r \right)\le t \frac{\iota}{2n}.
	\end{align*}
	%
	Let $Q^{\kvec,\theta,r} = \left\{ 1,\dots, j^* \right\}$ be the set of the $j^* = \left\lfloor\sum_{j \in[q]} x^{j,\kvec,\theta}_r \right\rfloor$ smallest indexes in $[q]$.
	Then, the algorithm updates the solution $x$ by setting $ x_r^{j,\kvec,\theta}=1$ for all indexes $j \in Q^{\kvec,\theta,r}$ and setting 
	\[
		 x_r^{j^*,\kvec,\theta}=\sum_{j' \in Q^{\kvec,\theta,r}} x^{j',\kvec,\theta}_r- \left\lfloor\sum_{j' \in Q^{\kvec,\theta,r}} x^{j',\kvec,\theta}_r \right\rfloor.
	\]
	%
	%
	
	After having iterated over all the receivers, the algorithm has built a new feasible solution $\bar x$ to Program~\ref{lp:multiSub} such that \[
		\sum_{j \in [q]} \left[ F_\theta(\bar x^{j,\kvec,\theta}_r)- F_\theta(x^{j,\kvec,\theta}_r) \right]\ge -q\iota/2,
	\]
	since the algorithm moved at most a value $q$ from variables indexed by $j'$ to variables indexed by$j<j'$.
	Moreover, each receiver $r \in \rec$ has at most a non-binary element among variables $\bar x^{j,\kvec,\theta}_r$.

	As a final step, the algorithm first builds a set $Q^{\kvec,\theta}$ of indexes $j \in [q]$ such that $\bar x^{j,\kvec,\theta}$ is a binary vector.
	Notice that there always exists one such set $Q^{\kvec,\theta}$ of size at least $q-n$.
	%
	%
	Then, the algorithm constructs a signaling scheme such that
	\[
		\phi^\kvec_\theta(R)=\frac{1}{q} \left|\left \{j \in Q^{\kvec,\theta} : x^{j,\kvec,\theta}_r=1\forall r \in R, x^{j,\kvec,\theta}_r=0\forall r \notin R\right\}\right|.
	\]
	Notice that $\sum_{R \in \rec: r \in R }\phi^\kvec_\theta(R)\le x^{r,k_r}$ and, by the monotonicity assumption on $f_\theta$, it is easy to build a signaling scheme such that $\sum_{R \in \rec: r \in R } \phi^\kvec_\theta(R)= x^{r,k_r}$ with greater sender's expected utility.
	Finally, the algorithm outputs the sender's strategy made by $\{\phi^\kvec\}_{\kvec \in \bar \K}$ and $\{x^{r,k}\}_{r \in \rec,k \in \K_r}$, where the menus of marginal signaling schemes are those given as input.

	To conclude the proof, we show that the  utility of the sender's strategy described above is at least $APX-\frac{n}{q}-\iota$ in expectation.
	If the event $\mathcal{E}$ holds, the utility of the solution is at least 
	\begin{align*}
		\sum_{\theta \in \Theta} \mu_\theta \sum_{\kvec \in \bar \K}  \lambda_\kvec &\phi^\kvec_\theta(R) f_\theta(R) \\
		&\ge \sum_{\theta \in \Theta} \mu_\theta \sum_{\kvec \in \bar \K}  \lambda_\kvec \frac{1}{q}\sum_{j \in Q^{\kvec,\theta}}  F_\theta \big(\bar x^{j,\kvec,\theta} \big)\\ 
		&\ge\sum_{\theta \in \Theta} \mu_\theta \sum_{\kvec \in \bar \K}  \lambda_\kvec \frac{1}{q}\left[\sum_{j \in [q]}  F_\theta \big(\bar x^{j,\kvec,\theta} \big)-n\right]\\
		&\ge \sum_{\theta \in \Theta} \mu_\theta \sum_{\kvec \in \bar \K}  \lambda_\kvec \frac{1}{q} \left[\sum_{j \in [q]}  F_\theta \big( x^{j,\kvec,\theta} \big)-n-\frac{\iota q}{2}\right]\\
		&= \sum_{\theta \in \Theta} \mu_\theta \sum_{\kvec \in \bar \K} \lambda_\kvec \left[\frac{1}{q}\sum_{j \in [q]}  F_\theta \big( x^{j,\kvec,\theta} \big)-\frac{n}{q}-\iota/2\right] \\
		& \ge \sum_{\theta \in \Theta} \mu_\theta \sum_{\kvec \in \bar \K } \lambda_\kvec \frac{1}{q}\sum_{j \in [q]}  F_\theta \big( x^{j,\kvec,\theta} \big)-\frac{n}{q}-\iota/2\\
		&=APX-\frac{n}{q}-\iota/2.
	\end{align*}
	Hence, the sender's expected utility is at least
	\begin{align*}
	\Pr \left\{ \mathcal{E} \right\} \left( APX - \frac{n}{q}-\iota/2 \right)&\ge (1-p |\bar \K| d q n) \left( APX-\frac{n}{q}-\frac{\iota}{2} \right)\\
	&\ge APX-\frac{n}{q}-\frac{\iota}{2}-p |\bar \K| d q n \\
	&\ge APX-\frac{n}{q}-\iota.
	\end{align*}
	Since we the marginal signaling schemes do not change, all the persuasiveness and IC constraints are satisfied.
	Moreover, for every $\kvec \in \bar \K,\theta \in \Theta$, and $r \in \rec$, it holds
	\[
		\sum_{R \subseteq \rec:r \in R} \phi^\kvec_\theta(R)=\frac{1}{q} \sum_{j \in [q]} \bar x^{j,\kvec,\theta}_r=\frac{1}{q} \sum_{j \in [q]}  x^{j,\kvec,\theta}_r=x^{r,k_r}_\theta,
	\]
	while it is easy to see that $\sum_{R \subseteq \rec} \phi^\kvec_\theta(R)=1$ for every $\kvec \in \bar \K$ and $\theta \in \Theta$.
	This concludes the proof of the lemma.
\end{proof}

\begin{algorithm}[!htp]
	\caption{Algorithm in Lemma \ref{lm:subDirect}}\label{alg:lemma3}
	\textbf{Input:} N. of samples $\sigma > 0$; Solution $x$ to Program~\ref{lp:multiSub}; $\kvec \in \bar \K$; $\theta \in \Theta$
	\begin{algorithmic}[1]
		\For{$r \in \rec$}
		\State {Compute $\tilde e^{j,\kvec,\theta}_r$ estimating $\frac{\partial F_\theta( x^{j,\kvec,\theta} )}{\partial x^{j,\kvec,\theta}_r}$  with $\sigma$ samples}
		\State{Re-label indexes $j \in [q]$ in decreasing order of $\tilde e^{j,\kvec,\theta}_r$}
		\State $j^*\gets \left\lfloor \sum_{j \in [q]} x^{j,\kvec,\theta}_r \right\rfloor$
		\State {$ x^{j^*+1,\kvec,\theta}_r \gets \sum_{j \in [q]} x^{j,\kvec,\theta}_r -j^*$}
		\State{ $Q^{\kvec,\theta,r}\gets \left\{ 1,\dots, j^* \right\}$}
		\For{$j \in Q^{\kvec,\theta,r}$}
		\State{ $x^{j,\kvec,\theta}_r \gets 1$}
		\EndFor
		\For{$j\ge j^*+2$}
		\State	{$x^{j,\kvec,\theta}_r\gets 0$ }
		\EndFor	 	
		\EndFor
		\State {Construct $\phi^\kvec_\theta$ as follows:\\
				$\phi^\kvec_\theta(R) \hspace{-0.05cm}=\hspace{-0.05cm}\frac{1}{q} \hspace{-0.05cm} \left|\left\{j \in\hspace{-0.02cm} [q] : \hspace{-0.05cm}x^{j,\kvec,\theta}_r\hspace{-0.05cm}=\hspace{-0.05cm}1\forall r \in \hspace{-0.05cm}R, x^{j,\kvec,\theta}_r\hspace{-0.05cm}=\hspace{-0.05cm}0\forall r \notin\hspace{-0.05cm} R\right\}\right| \forall R \hspace{-0.05cm}\subseteq\hspace{-0.05cm} \rec$}
		\State{Update $\phi^\kvec_\theta$ to make it consistent with the menus of marginal signaling schemes $\{x^{r,k_r}_\theta\}_{r \in \rec}$}
		\State\Return $\phi^{\kvec}_\theta$
	\end{algorithmic}
\end{algorithm}

Now, we can prove Theorem~\ref{thm:submodular}.

\begin{proof}[proof of Theorem \ref{thm:submodular}]
	 By Lemmas \ref{lm:subUniform} and \ref{lm:subDirect}, we only need to provide an algorithm that approximates the optimal solution of LP~\ref{lp:multiSub}.
	The objective is a linear combination with non-negative coefficients of the multi-linear extension of monotone submodular functions.
	Hence, it is smooth, monotone and submodular.
	Moreover, since we relaxed Constraints~\eqref{lp:multiRelaxed}, the feasible region is a down-monotone polytope\footnote{A polytope $\mathcal{P}\in \mathbb{R}^n_+$ is down-monotone if  $\mathbf{0}\le \mathbf{x} \le  \mathbf{y}$ coordinate-wise and $\mathbf{y} \in \mathcal{P}$ imply $\mathbf{x} \in \mathcal{P}$.} and it is defined by polynomially-many constraints.
	For each $\delta>0$, this problem admits a $\left( 1-\frac{1}{e} \right)OPT-\delta$-approximation in time polynomial in the instance size and $\delta$, see the continuous greedy algorithm in~\citep{CalinescuMaximizing2011}~and~\citep{dughmi2017algorithmic} for a formulation in a similar problem.\footnote{The bound holds only for arbitrary large probability. This reduces the total expected utility by an arbitrary small factor.}
	Finally, we can obtain an arbitrary good approximation choosing an arbitrary large value for $q$ and an arbitrary small value for $\delta$ and $\iota$.
\end{proof}

\section{Conclusions and Future Works}
We proposed to extend the Bayesian persuasion framework with a type reporting step.
We proved that, with a single receiver, the addition of this type reporting step makes the sender's computational problem tractable.
Moreover, we extended the framework to settings with multiple receivers, focusing on the widely-studied case with no inter-agent-externalities and binary actions.
We showed that an optimal sender's strategy can be computed in polynomial time when the sender's utility function is supermodular or anonymous. 
Moreover, when the sender's utility function is submodular, we designed a polynomial-time algorithm that provides a tight $\left(1-\frac{1}{e} \right)$-approximation.

In the future, it would be interesting to study the setting in which the sender has access only to samples from the distribution of the receiver's types.
Another interesting direction is to explore how the type reporting step can be used to provide polynomial-time no-regret algorithms in an online learning framework.

\begin{acks}
	This work has been partially supported by the Italian MIUR PRIN 2017 Project ALGADIMAR ``Algorithms, Games, and Digital Market''.
\end{acks}






\bibliographystyle{ACM-Reference-Format} 
\bibliography{refs}


\begin{thebibliography}{34}


\ifx \showCODEN    \undefined \def \showCODEN     #1{\unskip}     \fi
\ifx \showDOI      \undefined \def \showDOI       #1{#1}\fi
\ifx \showISBNx    \undefined \def \showISBNx     #1{\unskip}     \fi
\ifx \showISBNxiii \undefined \def \showISBNxiii  #1{\unskip}     \fi
\ifx \showISSN     \undefined \def \showISSN      #1{\unskip}     \fi
\ifx \showLCCN     \undefined \def \showLCCN      #1{\unskip}     \fi
\ifx \shownote     \undefined \def \shownote      #1{#1}          \fi
\ifx \showarticletitle \undefined \def \showarticletitle #1{#1}   \fi
\ifx \showURL      \undefined \def \showURL       {\relax}        \fi
\providecommand\bibfield[2]{#2}
\providecommand\bibinfo[2]{#2}
\providecommand\natexlab[1]{#1}
\providecommand\showeprint[2][]{arXiv:#2}

\bibitem[\protect\citeauthoryear{Alonso and C{\^a}mara}{Alonso and
  C{\^a}mara}{2016}]%
        {alonso2016persuading}
\bibfield{author}{\bibinfo{person}{Ricardo Alonso} {and}
  \bibinfo{person}{Odilon C{\^a}mara}.} \bibinfo{year}{2016}\natexlab{}.
\newblock \showarticletitle{Persuading voters}.
\newblock \bibinfo{journal}{\emph{American Economic Review}}
  \bibinfo{volume}{106}, \bibinfo{number}{11} (\bibinfo{year}{2016}),
  \bibinfo{pages}{3590--3605}.
\newblock


\bibitem[\protect\citeauthoryear{Arieli and Babichenko}{Arieli and
  Babichenko}{2019}]%
        {arieli2019private}
\bibfield{author}{\bibinfo{person}{Itai Arieli} {and} \bibinfo{person}{Yakov
  Babichenko}.} \bibinfo{year}{2019}\natexlab{}.
\newblock \showarticletitle{Private bayesian persuasion}.
\newblock \bibinfo{journal}{\emph{Journal of Economic Theory}}
  \bibinfo{volume}{182} (\bibinfo{year}{2019}), \bibinfo{pages}{185--217}.
\newblock


\bibitem[\protect\citeauthoryear{Babichenko and Barman}{Babichenko and
  Barman}{2017}]%
        {babichenko2017algorithmic}
\bibfield{author}{\bibinfo{person}{Y. Babichenko} {and} \bibinfo{person}{S.
  Barman}.} \bibinfo{year}{2017}\natexlab{}.
\newblock \showarticletitle{Algorithmic aspects of private bayesian
  persuasion}. In \bibinfo{booktitle}{\emph{Innovations in Theoretical Computer
  Science Conference}}.
\newblock


\bibitem[\protect\citeauthoryear{Babichenko, Talgam{-}Cohen, Xu, and
  Zabarnyi}{Babichenko et~al\mbox{.}}{2021}]%
        {babichenko2021regret}
\bibfield{author}{\bibinfo{person}{Yakov Babichenko}, \bibinfo{person}{Inbal
  Talgam{-}Cohen}, \bibinfo{person}{Haifeng Xu}, {and}
  \bibinfo{person}{Konstantin Zabarnyi}.} \bibinfo{year}{2021}\natexlab{}.
\newblock \showarticletitle{Regret-Minimizing Bayesian Persuasion}. In
  \bibinfo{booktitle}{\emph{{EC} '21: The 22nd {ACM} Conference on Economics
  and Computation, Budapest, Hungary, July 18-23, 2021}},
  \bibfield{editor}{\bibinfo{person}{P{\'{e}}ter Bir{\'{o}}},
  \bibinfo{person}{Shuchi Chawla}, {and} \bibinfo{person}{Federico Echenique}}
  (Eds.). \bibinfo{publisher}{{ACM}}, \bibinfo{pages}{128}.
\newblock
\urldef\tempurl%
\url{https://doi.org/10.1145/3465456.3467574}
\showDOI{\tempurl}


\bibitem[\protect\citeauthoryear{Bacchiocchi, Castiglioni, Marchesi, Romano,
  and Gatti}{Bacchiocchi et~al\mbox{.}}{2022}]%
        {bacchiocchi2022public}
\bibfield{author}{\bibinfo{person}{Francesco Bacchiocchi},
  \bibinfo{person}{Matteo Castiglioni}, \bibinfo{person}{Alberto Marchesi},
  \bibinfo{person}{Giulia Romano}, {and} \bibinfo{person}{Nicola Gatti}.}
  \bibinfo{year}{2022}\natexlab{}.
\newblock \bibinfo{title}{Public Signaling in Bayesian Ad Auctions}.
\newblock
\newblock
\showeprint[arxiv]{2201.09728}~[cs.GT]


\bibitem[\protect\citeauthoryear{Badanidiyuru, Bhawalkar, and Xu}{Badanidiyuru
  et~al\mbox{.}}{2018}]%
        {badanidiyuru2018targeting}
\bibfield{author}{\bibinfo{person}{Ashwinkumar Badanidiyuru},
  \bibinfo{person}{Kshipra Bhawalkar}, {and} \bibinfo{person}{Haifeng Xu}.}
  \bibinfo{year}{2018}\natexlab{}.
\newblock \showarticletitle{Targeting and signaling in ad auctions}. In
  \bibinfo{booktitle}{\emph{Proceedings of the Twenty-Ninth Annual ACM-SIAM
  Symposium on Discrete Algorithms}}. \bibinfo{pages}{2545--2563}.
\newblock


\bibitem[\protect\citeauthoryear{Bhaskar, Cheng, Ko, and Swamy}{Bhaskar
  et~al\mbox{.}}{2016}]%
        {bhaskar2016hardness}
\bibfield{author}{\bibinfo{person}{Umang Bhaskar}, \bibinfo{person}{Yu Cheng},
  \bibinfo{person}{Young~Kun Ko}, {and} \bibinfo{person}{Chaitanya Swamy}.}
  \bibinfo{year}{2016}\natexlab{}.
\newblock \showarticletitle{Hardness results for signaling in bayesian zero-sum
  and network routing games}. In \bibinfo{booktitle}{\emph{Proceedings of the
  2016 ACM Conference on Economics and Computation}}.
  \bibinfo{pages}{479--496}.
\newblock


\bibitem[\protect\citeauthoryear{Bro~Miltersen and Sheffet}{Bro~Miltersen and
  Sheffet}{2012}]%
        {bro2012send}
\bibfield{author}{\bibinfo{person}{Peter Bro~Miltersen} {and}
  \bibinfo{person}{Or Sheffet}.} \bibinfo{year}{2012}\natexlab{}.
\newblock \showarticletitle{Send mixed signals: earn more, work less}. In
  \bibinfo{booktitle}{\emph{Proceedings of the 13th ACM Conference on
  Electronic Commerce}}. \bibinfo{pages}{234--247}.
\newblock


\bibitem[\protect\citeauthoryear{Calinescu, Chekuri, Pál, and
  Vondrák}{Calinescu et~al\mbox{.}}{2011}]%
        {CalinescuMaximizing2011}
\bibfield{author}{\bibinfo{person}{Gruia Calinescu}, \bibinfo{person}{Chandra
  Chekuri}, \bibinfo{person}{Martin Pál}, {and} \bibinfo{person}{Jan
  Vondrák}.} \bibinfo{year}{2011}\natexlab{}.
\newblock \showarticletitle{Maximizing a Monotone Submodular Function Subject
  to a Matroid Constraint}.
\newblock \bibinfo{journal}{\emph{SIAM J. Comput.}} \bibinfo{volume}{40},
  \bibinfo{number}{6} (\bibinfo{year}{2011}), \bibinfo{pages}{1740--1766}.
\newblock
\urldef\tempurl%
\url{https://doi.org/10.1137/080733991}
\showDOI{\tempurl}
\showeprint{https://doi.org/10.1137/080733991}


\bibitem[\protect\citeauthoryear{Candogan}{Candogan}{2019}]%
        {candogan2019persuasion}
\bibfield{author}{\bibinfo{person}{Ozan Candogan}.}
  \bibinfo{year}{2019}\natexlab{}.
\newblock \showarticletitle{Persuasion in networks: Public signals and
  k-cores}. In \bibinfo{booktitle}{\emph{Proceedings of the 2019 ACM Conference
  on Economics and Computation}}. \bibinfo{pages}{133--134}.
\newblock


\bibitem[\protect\citeauthoryear{Castiglioni, Celli, and Gatti}{Castiglioni
  et~al\mbox{.}}{2020a}]%
        {castiglioni2019persuading}
\bibfield{author}{\bibinfo{person}{Matteo Castiglioni}, \bibinfo{person}{Andrea
  Celli}, {and} \bibinfo{person}{Nicola Gatti}.}
  \bibinfo{year}{2020}\natexlab{a}.
\newblock \showarticletitle{Persuading Voters: It's Easy to Whisper, It's Hard
  to Speak Loud}. In \bibinfo{booktitle}{\emph{The Thirty-Fourth AAAI
  Conference on Artificial Intelligence}}. \bibinfo{pages}{1870--1877}.
\newblock


\bibitem[\protect\citeauthoryear{Castiglioni, Celli, and Gatti}{Castiglioni
  et~al\mbox{.}}{2020b}]%
        {castiglioni2020public}
\bibfield{author}{\bibinfo{person}{Matteo Castiglioni}, \bibinfo{person}{Andrea
  Celli}, {and} \bibinfo{person}{Nicola Gatti}.}
  \bibinfo{year}{2020}\natexlab{b}.
\newblock \bibinfo{title}{Public Bayesian Persuasion: Being Almost Optimal and
  Almost Persuasive}.
\newblock
\newblock
\showeprint[arxiv]{2002.05156}~[cs.GT]


\bibitem[\protect\citeauthoryear{Castiglioni, Celli, Marchesi, and
  Gatti}{Castiglioni et~al\mbox{.}}{2020c}]%
        {Castiglioni2020online}
\bibfield{author}{\bibinfo{person}{Matteo Castiglioni}, \bibinfo{person}{Andrea
  Celli}, \bibinfo{person}{Alberto Marchesi}, {and} \bibinfo{person}{Nicola
  Gatti}.} \bibinfo{year}{2020}\natexlab{c}.
\newblock \showarticletitle{Online Bayesian Persuasion}. In
  \bibinfo{booktitle}{\emph{Advances in Neural Information Processing
  Systems}}, \bibfield{editor}{\bibinfo{person}{H.~Larochelle},
  \bibinfo{person}{M.~Ranzato}, \bibinfo{person}{R.~Hadsell},
  \bibinfo{person}{M.~F. Balcan}, {and} \bibinfo{person}{H.~Lin}} (Eds.),
  Vol.~\bibinfo{volume}{33}. \bibinfo{publisher}{Curran Associates, Inc.},
  \bibinfo{pages}{16188--16198}.
\newblock
\urldef\tempurl%
\url{https://proceedings.neurips.cc/paper/2020/file/ba5451d3c91a0f982f103cdbe249bc78-Paper.pdf}
\showURL{%
\tempurl}


\bibitem[\protect\citeauthoryear{Castiglioni, Celli, Marchesi, and
  Gatti}{Castiglioni et~al\mbox{.}}{2021a}]%
        {castiglioni2020signaling}
\bibfield{author}{\bibinfo{person}{Matteo Castiglioni}, \bibinfo{person}{Andrea
  Celli}, \bibinfo{person}{Alberto Marchesi}, {and} \bibinfo{person}{Nicola
  Gatti}.} \bibinfo{year}{2021}\natexlab{a}.
\newblock \showarticletitle{Signaling in Bayesian Network Congestion Games: the
  Subtle Power of Symmetry}. In \bibinfo{booktitle}{\emph{The Thirty-Fifth AAAI
  Conference on Artificial Intelligence}}.
\newblock


\bibitem[\protect\citeauthoryear{Castiglioni and Gatti}{Castiglioni and
  Gatti}{2021}]%
        {castiglioni2020persuading}
\bibfield{author}{\bibinfo{person}{Matteo Castiglioni} {and}
  \bibinfo{person}{Nicola Gatti}.} \bibinfo{year}{2021}\natexlab{}.
\newblock \showarticletitle{Persuading Voters in District-based Elections}. In
  \bibinfo{booktitle}{\emph{The Thirty-Fifth AAAI Conference on Artificial
  Intelligence}}.
\newblock


\bibitem[\protect\citeauthoryear{Castiglioni, Marchesi, Celli, and
  Gatti}{Castiglioni et~al\mbox{.}}{2021b}]%
        {castiglioniMulti2021}
\bibfield{author}{\bibinfo{person}{Matteo Castiglioni},
  \bibinfo{person}{Alberto Marchesi}, \bibinfo{person}{Andrea Celli}, {and}
  \bibinfo{person}{Nicola Gatti}.} \bibinfo{year}{2021}\natexlab{b}.
\newblock \showarticletitle{Multi-Receiver Online Bayesian Persuasion}. In
  \bibinfo{booktitle}{\emph{Proceedings of the 38th International Conference on
  Machine Learning, {ICML} 2021, 18-24 July 2021, Virtual Event}}
  \emph{(\bibinfo{series}{Proceedings of Machine Learning Research},
  Vol.~\bibinfo{volume}{139})}, \bibfield{editor}{\bibinfo{person}{Marina
  Meila} {and} \bibinfo{person}{Tong Zhang}} (Eds.).
  \bibinfo{publisher}{{PMLR}}, \bibinfo{pages}{1314--1323}.
\newblock
\urldef\tempurl%
\url{http://proceedings.mlr.press/v139/castiglioni21a.html}
\showURL{%
\tempurl}


\bibitem[\protect\citeauthoryear{Castiglioni, Romano, Marchesi, and
  Gatti}{Castiglioni et~al\mbox{.}}{2022}]%
        {castiglioni2022signaling}
\bibfield{author}{\bibinfo{person}{Matteo Castiglioni}, \bibinfo{person}{Giulia
  Romano}, \bibinfo{person}{Alberto Marchesi}, {and} \bibinfo{person}{Nicola
  Gatti}.} \bibinfo{year}{2022}\natexlab{}.
\newblock \bibinfo{title}{Signaling in Posted Price Auctions}.
\newblock
\newblock
\showeprint[arxiv]{2201.12183}~[cs.GT]


\bibitem[\protect\citeauthoryear{Cheng, Cheung, Dughmi, Emamjomeh-Zadeh, Han,
  and Teng}{Cheng et~al\mbox{.}}{2015}]%
        {cheng2015mixture}
\bibfield{author}{\bibinfo{person}{Yu Cheng}, \bibinfo{person}{Ho~Yee Cheung},
  \bibinfo{person}{Shaddin Dughmi}, \bibinfo{person}{Ehsan Emamjomeh-Zadeh},
  \bibinfo{person}{Li Han}, {and} \bibinfo{person}{Shang-Hua Teng}.}
  \bibinfo{year}{2015}\natexlab{}.
\newblock \showarticletitle{Mixture Selection, Mechanism Design, and
  Signaling}. In \bibinfo{booktitle}{\emph{56th Annual Symposium on Foundations
  of Computer Science}}. \bibinfo{pages}{1426--1445}.
\newblock


\bibitem[\protect\citeauthoryear{Conitzer and Sandholm}{Conitzer and
  Sandholm}{2002}]%
        {conitzer2002complexity}
\bibfield{author}{\bibinfo{person}{Vincent Conitzer} {and}
  \bibinfo{person}{Tuomas Sandholm}.} \bibinfo{year}{2002}\natexlab{}.
\newblock \showarticletitle{Complexity of mechanism design}. In
  \bibinfo{booktitle}{\emph{Proceedings of the Eighteenth conference on
  Uncertainty in artificial intelligence}}. \bibinfo{pages}{103--110}.
\newblock


\bibitem[\protect\citeauthoryear{Conitzer and Sandholm}{Conitzer and
  Sandholm}{2003}]%
        {conitzer2003automated}
\bibfield{author}{\bibinfo{person}{Vincent Conitzer} {and}
  \bibinfo{person}{Tuomas Sandholm}.} \bibinfo{year}{2003}\natexlab{}.
\newblock \showarticletitle{Automated mechanism design: Complexity results
  stemming from the single-agent setting}. In
  \bibinfo{booktitle}{\emph{Proceedings of the 5th international conference on
  Electronic commerce}}. \bibinfo{pages}{17--24}.
\newblock


\bibitem[\protect\citeauthoryear{Dughmi and Xu}{Dughmi and Xu}{2016}]%
        {dughmi2016algorithmic}
\bibfield{author}{\bibinfo{person}{S. Dughmi} {and} \bibinfo{person}{H. Xu}.}
  \bibinfo{year}{2016}\natexlab{}.
\newblock \showarticletitle{Algorithmic bayesian persuasion}. In
  \bibinfo{booktitle}{\emph{ACM STOC}}. \bibinfo{pages}{412--425}.
\newblock


\bibitem[\protect\citeauthoryear{Dughmi and Xu}{Dughmi and Xu}{2017}]%
        {dughmi2017algorithmic}
\bibfield{author}{\bibinfo{person}{S. Dughmi} {and} \bibinfo{person}{H. Xu}.}
  \bibinfo{year}{2017}\natexlab{}.
\newblock \showarticletitle{Algorithmic persuasion with no externalities}. In
  \bibinfo{booktitle}{\emph{ACM EC}}. \bibinfo{pages}{351--368}.
\newblock


\bibitem[\protect\citeauthoryear{Emek, Feldman, Gamzu, PaesLeme, and
  Tennenholtz}{Emek et~al\mbox{.}}{2014}]%
        {emek2014signaling}
\bibfield{author}{\bibinfo{person}{Yuval Emek}, \bibinfo{person}{Michal
  Feldman}, \bibinfo{person}{Iftah Gamzu}, \bibinfo{person}{Renato PaesLeme},
  {and} \bibinfo{person}{Moshe Tennenholtz}.} \bibinfo{year}{2014}\natexlab{}.
\newblock \showarticletitle{Signaling schemes for revenue maximization}.
\newblock \bibinfo{journal}{\emph{ACM Transactions on Economics and
  Computation}} \bibinfo{volume}{2}, \bibinfo{number}{2}
  (\bibinfo{year}{2014}), \bibinfo{pages}{1--19}.
\newblock


\bibitem[\protect\citeauthoryear{Guo and Conitzer}{Guo and Conitzer}{2010}]%
        {guo2010computationally}
\bibfield{author}{\bibinfo{person}{Mingyu Guo} {and} \bibinfo{person}{Vincent
  Conitzer}.} \bibinfo{year}{2010}\natexlab{}.
\newblock \showarticletitle{Computationally feasible automated mechanism
  design: General approach and case studies}. In
  \bibinfo{booktitle}{\emph{Proceedings of the AAAI Conference on Artificial
  Intelligence}}, Vol.~\bibinfo{volume}{24}.
\newblock


\bibitem[\protect\citeauthoryear{Kamenica}{Kamenica}{2019}]%
        {Kamenica2019Bayesian}
\bibfield{author}{\bibinfo{person}{Emir Kamenica}.}
  \bibinfo{year}{2019}\natexlab{}.
\newblock \showarticletitle{Bayesian Persuasion and Information Design}.
\newblock \bibinfo{journal}{\emph{Annual Review of Economics}}
  \bibinfo{volume}{11}, \bibinfo{number}{1} (\bibinfo{year}{2019}),
  \bibinfo{pages}{249--272}.
\newblock
\urldef\tempurl%
\url{https://doi.org/10.1146/annurev-economics-080218-025739}
\showDOI{\tempurl}
\showeprint{https://doi.org/10.1146/annurev-economics-080218-025739}


\bibitem[\protect\citeauthoryear{Kamenica and Gentzkow}{Kamenica and
  Gentzkow}{2011}]%
        {kamenica2011bayesian}
\bibfield{author}{\bibinfo{person}{Emir Kamenica} {and}
  \bibinfo{person}{Matthew Gentzkow}.} \bibinfo{year}{2011}\natexlab{}.
\newblock \showarticletitle{Bayesian persuasion}.
\newblock \bibinfo{journal}{\emph{American Economic Review}}
  \bibinfo{volume}{101}, \bibinfo{number}{6} (\bibinfo{year}{2011}),
  \bibinfo{pages}{2590--2615}.
\newblock


\bibitem[\protect\citeauthoryear{Mansour, Slivkins, Syrgkanis, and Wu}{Mansour
  et~al\mbox{.}}{2016}]%
        {mansour2016bayesian}
\bibfield{author}{\bibinfo{person}{Yishay Mansour}, \bibinfo{person}{Aleksandrs
  Slivkins}, \bibinfo{person}{Vasilis Syrgkanis}, {and}
  \bibinfo{person}{Zhiwei~Steven Wu}.} \bibinfo{year}{2016}\natexlab{}.
\newblock \showarticletitle{Bayesian Exploration: Incentivizing Exploration in
  Bayesian Games}. In \bibinfo{booktitle}{\emph{Proceedings of the 2016 ACM
  Conference on Economics and Computation}}. \bibinfo{pages}{661--661}.
\newblock


\bibitem[\protect\citeauthoryear{Rabinovich, Jiang, Jain, and Xu}{Rabinovich
  et~al\mbox{.}}{2015}]%
        {rabinovich2015information}
\bibfield{author}{\bibinfo{person}{Zinovi Rabinovich},
  \bibinfo{person}{Albert~Xin Jiang}, \bibinfo{person}{Manish Jain}, {and}
  \bibinfo{person}{Haifeng Xu}.} \bibinfo{year}{2015}\natexlab{}.
\newblock \showarticletitle{Information disclosure as a means to security}. In
  \bibinfo{booktitle}{\emph{Proceedings of the 2015 International Conference on
  Autonomous Agents and Multiagent Systems}}. \bibinfo{pages}{645--653}.
\newblock


\bibitem[\protect\citeauthoryear{Shoham and Leyton-Brown}{Shoham and
  Leyton-Brown}{2008}]%
        {shoham2008multiagent}
\bibfield{author}{\bibinfo{person}{Yoav Shoham} {and} \bibinfo{person}{Kevin
  Leyton-Brown}.} \bibinfo{year}{2008}\natexlab{}.
\newblock \bibinfo{booktitle}{\emph{Multiagent systems: Algorithmic,
  game-theoretic, and logical foundations}}.
\newblock \bibinfo{publisher}{Cambridge University Press}.
\newblock


\bibitem[\protect\citeauthoryear{Vasserman, Feldman, and Hassidim}{Vasserman
  et~al\mbox{.}}{2015}]%
        {vasserman2015implementing}
\bibfield{author}{\bibinfo{person}{Shoshana Vasserman}, \bibinfo{person}{Michal
  Feldman}, {and} \bibinfo{person}{Avinatan Hassidim}.}
  \bibinfo{year}{2015}\natexlab{}.
\newblock \showarticletitle{Implementing the wisdom of waze}. In
  \bibinfo{booktitle}{\emph{Twenty-Fourth International Joint Conference on
  Artificial Intelligence}}. \bibinfo{pages}{660--666}.
\newblock


\bibitem[\protect\citeauthoryear{Vorobeychik, Kiekintveld, and
  Wellman}{Vorobeychik et~al\mbox{.}}{2006}]%
        {vorobeychik2006empirical}
\bibfield{author}{\bibinfo{person}{Yevgeniy Vorobeychik},
  \bibinfo{person}{Christopher Kiekintveld}, {and} \bibinfo{person}{Michael~P
  Wellman}.} \bibinfo{year}{2006}\natexlab{}.
\newblock \showarticletitle{Empirical mechanism design: Methods, with
  application to a supply-chain scenario}. In
  \bibinfo{booktitle}{\emph{Proceedings of the 7th ACM conference on Electronic
  commerce}}. \bibinfo{pages}{306--315}.
\newblock


\bibitem[\protect\citeauthoryear{Xu}{Xu}{2020}]%
        {Xu2020tractability}
\bibfield{author}{\bibinfo{person}{Haifeng Xu}.}
  \bibinfo{year}{2020}\natexlab{}.
\newblock \showarticletitle{On the Tractability of Public Persuasion with No
  Externalities}. In \bibinfo{booktitle}{\emph{Proceedings of the 2020 ACM-SIAM
  Symposium on Discrete Algorithms (SODA)}}. \bibinfo{pages}{2708--2727}.
\newblock


\bibitem[\protect\citeauthoryear{Xu, Freeman, Conitzer, Dughmi, and Tambe}{Xu
  et~al\mbox{.}}{2016}]%
        {xu2016signaling}
\bibfield{author}{\bibinfo{person}{Haifeng Xu}, \bibinfo{person}{Rupert
  Freeman}, \bibinfo{person}{Vincent Conitzer}, \bibinfo{person}{Shaddin
  Dughmi}, {and} \bibinfo{person}{Milind Tambe}.}
  \bibinfo{year}{2016}\natexlab{}.
\newblock \showarticletitle{Signaling in Bayesian Stackelberg Games}. In
  \bibinfo{booktitle}{\emph{Proceedings of the 2016 International Conference on
  Autonomous Agents and Multiagent Systems}}. \bibinfo{pages}{150--158}.
\newblock


\bibitem[\protect\citeauthoryear{Zu, Iyer, and Xu}{Zu et~al\mbox{.}}{2021}]%
        {zu2021Learning}
\bibfield{author}{\bibinfo{person}{You Zu}, \bibinfo{person}{Krishnamurthy
  Iyer}, {and} \bibinfo{person}{Haifeng Xu}.} \bibinfo{year}{2021}\natexlab{}.
\newblock \showarticletitle{Learning to Persuade on the Fly: Robustness Against
  Ignorance}. In \bibinfo{booktitle}{\emph{{EC} '21: The 22nd {ACM} Conference
  on Economics and Computation, Budapest, Hungary, July 18-23, 2021}},
  \bibfield{editor}{\bibinfo{person}{P{\'{e}}ter Bir{\'{o}}},
  \bibinfo{person}{Shuchi Chawla}, {and} \bibinfo{person}{Federico Echenique}}
  (Eds.). \bibinfo{publisher}{{ACM}}, \bibinfo{pages}{927--928}.
\newblock
\urldef\tempurl%
\url{https://doi.org/10.1145/3465456.3467593}
\showDOI{\tempurl}


\end{thebibliography}


\clearpage
\onecolumn
\appendix

\section{Proofs Omitted from Section~\ref{sec:results_single}}

\lemmaone*

 \begin{proof}
	We show that, given a menu of signaling schemes $\Phi = \{ \phi^k \}_{k \in K}$ with each $\phi^k$ encoded as a probability distribution $\gamma^k \in \Delta_{\Xi}$, we can construct a new menu of signaling schemes $\bar \Phi = \{ \bar \phi^k \}_{k \in K}$ with each $\bar \phi^k$ encoded as a finite-supported probability distribution $\bar{{\gamma}}^k \in \Delta_{\pset^*}$ and such that the sender's expected utility for $\bar \Phi$ is greater than or equal to that for $\Phi$.
	This immediately proves the statement.

	In order to do so, we split the posteriors in $\pset$ into the sets $\hat \pset^{\avec}$ for $\avec \in \bigtimes_{k \in K} A$.
	Notice that $\pset=\bigcup_{\avec \in \bigtimes_{k \in K} A } \hat \pset^{\avec}$.
	Then, we replace the distributions $\gamma^k$ with other probability distributions supported on sets $V(\pset^\avec)\subseteq \pset^*$.
	For every action profile $\avec \in \bigtimes_{k \in K} A$ and type $k \in K$, we let $\pvec^{k,\avec} \coloneqq \mathbb{E}_{\pvec \sim \gamma^k} \left[ \pvec  \mid \pvec \in \hat \pset^\avec \right]$.
	Since $\hat \pset^\avec \subseteq \pset^\avec$ and $\pset^\avec$ is a bounded convex polytope, by Carathèodory theorem there exists a probability distribution $ {\gamma}^{k,\avec}\in \Delta_{\pset^*}$ such that its support is a subset of the set of vertices $V(\pset^{\avec})$ and it holds $\mathbb{E}_{\pvec\sim {\gamma}^{k,\avec}} \left[ \pvec \right]= \pvec^{k,\avec}$. 
	Then, let us define the probability distributions $\bar{{\gamma}}^k \in \Delta_{\pset^*}$ for $k \in K$ so that, for every posterior $\pvec \in \pset^*$, it holds
	\[
		\bar \gamma^k_\pvec = \sum_{\avec \in \bigtimes_{k \in K} A} \gamma^{k,\avec}_\pvec \, \text{Pr}_{\pvec' \sim \gamma^k} \left\{ \pvec' \in \hat \pset^\avec \right\}.
	\]
	
	Next, we show that the distributions $\bar{{\gamma}}^k \in \Delta_{\pset^*}$ for $k \in K$ defined above constitute a feasible solution to LP~\ref{lp:GenFinite} and the sender's expected utility in the resulting menu of signaling schemes $\bar \Phi$ is at least as large as the sender's expected utility for the menu of signaling schemes $\Phi$.
	%
	First, 
	let us notice that,
	for every $\avec \in \bigtimes_{k \in K} A$, $k \in K$, and $k' \in K$, it holds
	\begin{equation}\label{eq:distToExp}
		\sum_{\pvec \in V(\pset^\avec)} \gamma^{k',\avec}_\pvec \left[ \sum_{\theta \in \Theta} \p_\theta u^k_\theta \left( b^k_\pvec \right) \right] = \sum_{\pvec \in V(\pset^\avec)} \gamma^{k',\avec}_\pvec \left[ \sum_{\theta \in \Theta} \p_\theta u^k_\theta (a_k) \right] = 
		\sum_{\theta \in \Theta} \p^{k',\avec}_\theta u^k_\theta(a_k) =\mathbb{E}_{\pvec\sim \gamma^{k'}} \left[ \sum_{\theta \in \Theta} \p_\theta u^k_\theta(a_k)  \mid \pvec \in \hat \pset^\avec \right],
	\end{equation}
	where the second equality comes from the fact that action $a_k$ is the best response of the receiver of type $k$ in each posterior $\pvec \in\pset^\avec$.
	Similarly, we can prove that, for every $\avec \in \bigtimes_{k \in K} A$ and $k \in K$, it holds
	\begin{equation}\label{eq:distToExp_send}
		\sum_{\pvec \in V(\pset^\avec)} \gamma^{k,\avec}_{\pvec}  \left[ \sum_{\theta \in \Theta} \p_\theta u^\mathsf{s}_\theta \left( b^k_\pvec \right)\right] \ge \sum_{\pvec \in V(\pset^\avec)} \gamma^{k,\avec}_{\pvec} \left[ \sum_{\theta \in \Theta} \p_\theta u^\mathsf{s}_\theta(a_k)  \right] = 
		\sum_{\theta \in \Theta} \p^{k,\avec}_\theta u^\mathsf{s}_\theta(a_k)   = \mathbb{E}_{\pvec\sim \gamma^k} \left[ \sum_{\theta \in \Theta} \p_\theta u^\mathsf{s}_\theta(a_k)  \mid \pvec \in \hat \pset^\avec \right]. 
	\end{equation}
	Then, we can show that the IC constraints, namely Constraints~\eqref{lp:GenFinite1}, are satisfied.
	Formally, for every $k \in K$ and $k' \in K : k \neq k'$, we have:
	\begin{align*}
		\sum_{\pvec \in \pset^*} \bar \gamma^{k'}_{\pvec} \sum_{\theta \in \Theta} \p_\theta u_\theta^k\left( b^k_\pvec \right) & = \sum_{\pvec \in \pset^*} \sum_{\avec \in \bigtimes_{k \in K} A} \gamma^{k',\avec}_\pvec \, \text{Pr}_{\pvec'\sim \gamma^{k'}} \left\{ \pvec' \in \hat \pset^\avec \right\} \sum_{\theta \in \Theta} \p_\theta u^k_\theta \left( b^k_\pvec \right) \\
		& = \sum_{\avec \in \bigtimes_{k \in K} A} \text{Pr}_{\pvec'\sim \gamma^{k'}} \left\{ \pvec' \in \hat \pset^\avec \right\} \sum_{\pvec \in V(\pset^\avec)} \gamma^{k',\avec}_\pvec \left[ \sum_{\theta \in \Theta} \p_\theta u^k_\theta \left( b^k_\pvec \right) \right] \\
		&= \sum_{\avec \in \bigtimes_{k \in K} A} \text{Pr}_{\pvec'\sim \gamma^k} \left\{  \pvec' \in \hat \pset^\avec \right\} \mathbb{E}_{\pvec\sim \gamma^{k'}} \left[ \sum_{\theta \in \Theta} \p_\theta u^k_\theta(a_k)  \mid \pvec \in \hat \pset^\avec \right]  \\
		& = \mathbb{E}_{\pvec\sim \gamma^{k'}} \left[ \sum_{\theta \in \Theta} \p_\theta u^k_\theta \left( b^k_\pvec \right) \right],
	\end{align*}
	where the second equality comes from the fact that $\gamma^{k',\avec}_\pvec$ is non-zero only for posteriors $\pvec\in V(\pset^\avec)$ and in third equality we use Equation~\eqref{eq:distToExp}.
	Hence, for every $k \in K$ and $k' \in K : k \neq k'$, we have
	\[
		\sum_{\pvec \in \pset^*} \bar \gamma^{k}_{\pvec}  \sum_{\theta \in  \Theta} \p_\theta u_\theta^k \left( b^k_\pvec \right) = \mathbb{E}_{\pvec\sim \gamma^k} \left[ \sum_{\theta \in \Theta} \p_\theta u^k \left( b^k_\pvec \right) \right] \geq \mathbb{E}_{\pvec\sim \gamma^{k'}} \left[ \sum_{\theta \in \Theta} \p_\theta u^k \left( b^k_\pvec \right) \right] =\sum_{\pvec \in \pset^*} \gamma^{k'}_{\pvec} \sum_{\theta \in \Theta} \p_\theta u_\theta^k \left( b^k_{\pvec} \right),
	\]
	where the inequality comes from the w.l.o.g. assumption that the menu $\Phi$ is IC.
	This proves that Constraints~\eqref{lp:GenFinite1} hold.
	Similarly, we can prove that the sender's expected utility does \emph{not} decrease when using $\bar \Phi$ rather than $\Phi$.
	Formally,
	\begin{align*}	
		\sum_{k \in K} \lambda_k \sum_{\pvec \in \pset^*} \bar \gamma^k_{\pvec} \sum_{\theta \in \Theta} \p_\theta u^\mathsf{s}_\theta \left( b^k_\pvec \right) & =\sum_{k \in K} \lambda_k \sum_{\pvec \in \pset^*} \sum_{\avec \in \bigtimes_{k \in K} A} \gamma^{k',\avec}_\pvec \,\text{Pr}_{\pvec'\sim \gamma^{k}} \left\{  \pvec' \in \hat \pset^\avec \right\} \sum_{\theta \in \Theta} \p_\theta u^\mathsf{s}_\theta \left( b^k_\pvec \right) \\
		& =\sum_{k \in K} \lambda_k  \sum_{\avec \in \bigtimes_{k \in K} A} \text{Pr}_{\pvec'\sim \gamma^{k}}\left\{  \pvec' \in \hat \pset^a  \right\}\sum_{\pvec \in V(\pset^{\avec})}  \gamma^{k',\avec}_\pvec \sum_{\theta \in \Theta} \p_\theta u^\mathsf{s}_\theta \left( b^k_\pvec \right) \\
		&  \ge \sum_{k \in K} \lambda_k \sum_{\avec \in \bigtimes_{k \in K} A} \text{Pr}_{\pvec'\sim \gamma^{k}} \left\{ \pvec' \in \hat \pset^\avec \right\} \mathbb{E}_{\pvec\sim \gamma^k} \left[ \sum_{\theta \in \Theta} \p_\theta u^\mathsf{s}_\theta(a_k)  \mid \pvec \in \hat \pset^\avec \right]  \\
		& = \sum_{k \in K} \lambda_k \mathbb{E}_{\pvec\sim \gamma^k}   \left[ \sum_{\theta \in \Theta} \p_\theta u^\mathsf{s}_\theta \left(b^k_{\pvec} \right) \right],
	\end{align*}
	where the inequality comes from Equation~\eqref{eq:distToExp_send}.
	Moreover, Constraints~\eqref{lp:GenFinite2} are satisfied, since
	\begin{align*}
		\sum_{\pvec \in \pset^*} \bar \gamma^k_\pvec \p_\theta & = \sum_{\pvec \in \pset^*}\sum_{\avec \in \bigtimes_{k \in K} A} \gamma^{k,\avec}_\pvec \, \text{Pr}_{\pvec'\sim \gamma^k} \left\{ \pvec' \in \hat \pset^\avec \right\} \p_\theta  \\
		& = \sum_{\avec \in \bigtimes_{k \in K} A} \text{Pr}_{\pvec'\sim \gamma^k} \left\{ \pvec' \in \hat \pset^\avec \right\} \sum_{\pvec \in V(\pset^\avec)}  \gamma^{k,\avec}_\pvec  \p_\theta \\ 
		& = \sum_{\avec \in \bigtimes_{k \in K} A} \text{Pr}_{\pvec'\sim \gamma^k} \left\{ \pvec' \in \hat \pset^\avec \right\} \mathbb{E}_{\pvec \sim \gamma^k} \left[  \p_\theta \mid \pvec \in \hat \pset^\avec \right]\\
		& =\mathbb{E}_{\pvec \sim \gamma^k} \left[ \p_\theta \right]=\mu_\theta.
	\end{align*}
	Finally, it is easy to see that the $\bar{\boldsymbol{\gamma}}^k$ are valid probability distributions.
	Indeed, for every $k \in K$, it holds
	\[
		\sum_{\pvec \in \pset^*} \bar \gamma^k_\pvec =\sum_{\avec \in \bigtimes_{k \in K} A} \text{Pr}_{\pvec'\sim \gamma^k} \left\{  \pvec' \in \hat \pset^\avec \right\} \sum_{\pvec \in \pset^*} \gamma^{k,\avec}_{\pvec} =  \sum_{\avec \in \bigtimes_{k \in K} A} \text{Pr}_{\pvec'\sim \gamma^k} \left\{  \pvec' \in \hat \pset^\avec \right\}=1.
	\]
	This concludes the proof.
\end{proof}

\theoremSingle*

\begin{proof}
	Since LP~\ref{lp:direct} has polynomially-many variables and constraints, an optimal menu of direct and persuasive signaling schemes can be computed in polynomial time by solving the LP.
	Thus, we only need to show that, in any single-receiver instance, there always exists an optimal menu of direct and persuasive signaling schemes.
	In particular, we show that, given an optimal solution $\{\boldsymbol{\gamma}^k\}_{k \in K}$ to LP~\ref{lp:GenFinite}, there exists a solution to LP~\ref{lp:direct} with the same value.
	The menu $\Phi = \{\phi^k\}_{k \in K}$ of signaling schemes defined by the solution to LP~\ref{lp:direct} is the desired optimal menu of direct and persuasive signaling schemes.
	We define the solution to LP~\ref{lp:direct} as follows.
	For every $k \in K$, $a \in A$, and $\theta \in \Theta$, we let $\phi^k_\theta(a) = \frac{\sum_{\pvec \in \hat \pset^{k,a} \cap \pset^*} \gamma^k_\pvec \p_\theta}{\mu_\theta}$.
	First, we prove that the two solutions have the same objective value.
	Formally,
	\[
		\sum_{k \in K} \lambda_k \sum_{\theta \in \Theta} \mu_\theta \sum_{a \in A} \phi^k_\theta(a) u_\theta^\mathsf{s}(a) =  \sum_{k \in K} \lambda_k \sum_{\theta \in \Theta} \sum_{a \in A}  \sum_{\pvec \in \hat \pset^{k,a} \cap \pset^*} \gamma^k_\pvec \p_\theta u_\theta^\mathsf{s}(a) =  \sum_{k \in K} \lambda_k \sum_{\pvec \in \pset^*} \gamma^k_{\pvec} \sum_{\theta \in \Theta} \p_\theta u^\mathsf{s}_\theta \left( b^k_\pvec \right),
	\]
	where the last equality follows from the fact that $b^k_\pvec=a$ for all the posteriors in $\pvec \in \hat \pset^{k,a}$.
	Thus, we are left to check that the solution is feasible.
	Recall that Constraints~\eqref{lp:direct1}~and~\eqref{lp:direct2} are equivalent to the constraints in Equation~\eqref{eq:ICDirect}.
	The latter are satisfied since, for every $k\neq k' \in K$, it holds
	\begin{align*}
		\sum_{a \in A}  \sum_{\theta \in \Theta} \mu_\theta \phi_\theta^k(a) u^{k}_\theta (a) & = \sum_{a \in A}  \sum_{\theta \in \Theta} \sum_{\pvec \in \hat \pset^{k,a}\cap \pset^*} \gamma^k_\pvec \p_\theta u^{k}_\theta (a) \\
		& =\sum_{\pvec \in \pset^*} \gamma^{k}_{\pvec} \sum_{\theta \in \Theta} \p_\theta u_\theta^k \left( b^k_\pvec \right) \\
		& \ge  \sum_{\pvec \in \pset^*} \gamma^{k'}_{\pvec} \sum_{\theta \in \Theta} \p_\theta u_\theta^k \left( b^k_{\pvec} \right) \\
		& =\sum_{a \in A} \sum_{\pvec \in \hat \pset^{k,a} \cap \pset^*} \gamma^{k'}_{\pvec} \sum_{\theta\in \Theta} \p_\theta u_\theta^k \left( b^k_\pvec \right) \\
		& = \sum_{a \in A} \sum_{\pvec \in \hat \pset^{k,a}} \gamma^{k'}_{\pvec} \max_{a' \in A}\sum_{\theta \in \Theta} \p_\theta u_\theta^k(a') \\
		& \ge \sum_{a \in A} \max_{a' \in A} \sum_{\pvec \in \hat \pset^{k,a}\cap \pset^*} \gamma^{k'}_{\pvec} \sum_{\theta \in \Theta} \p_\theta u_\theta^k(a') \\
		& =\sum_{a \in A} \max_{a' \in A} \sum_{\theta \in \Theta} \mu_\theta \phi^{k'}_\theta(a) u^k_\theta(a').
	\end{align*}
	Moreover, each signaling scheme $\phi^k$ is persuasive, since, for every $k \in K$, and $a\neq a' \in A$, it holds
	\begin{align*}
	\sum_{\theta \in \Theta} \mu_\theta \phi^k_\theta(a) u^{k}_\theta (a) &= \sum_{\theta \in \Theta} \sum_{\pvec \in \hat \pset^{k,a}\cap \pset^*} \gamma^k_\pvec \p_\theta u^{k}_\theta (a) \\
	& =\sum_{\theta \in \Theta} \sum_{\pvec \in \hat \pset^{k,a}\cap \pset^*} \gamma^k_\pvec \p_\theta u^{k}_\theta \left( b^k_\pvec \right)\\ 
	& \ge \sum_{\theta \in \Theta} \sum_{\pvec \in \hat \pset^{k,a}\cap \pset^*} \gamma^k_\pvec \p_\theta u^{k}_\theta (a') \\
	& =\sum_{\theta \in \Theta} \mu_\theta \phi^k_\theta(a) u^{k}_\theta (a'),
	\end{align*}
	and it is well defined since, for every $k \in K$ and $\theta \in \Theta$, it holds
	\[ 
		\sum_{a \in A} \phi^k_\theta(a) = \sum_{a \in A} \sum_{\pvec \in \hat \pset^{k,a}\cap \pset^*} \frac{\gamma^k_\pvec \p_\theta}{\mu_\theta} = \sum_{\pvec\in \pset^*} \frac{\gamma^k_\pvec \p_\theta}{\mu_\theta} =\frac{\mu_\theta}{\mu_\theta}=1.
	\]
	This concludes the proof.
	
	
\end{proof}

\section{Proofs Omitted from Section~\ref{sec:results_multiSup}}

\theoremDirectMulti*

\begin{proof}
	The key insight of the proof is that, in a multi-receiver instance, the sender's expected utility only depends on the marginal probabilities with which the receivers play actions $a_1$ and $a_0$ given each state of nature.
	In order to see that, observe that, once the marginal probabilities $ x^{r,k}_\theta$ are fixed, an optimal (general) signaling scheme $\phi^\kvec$ can be computed by solving LP~\ref{lp:multi} with Constraints~\eqref{lp:multi1}~and~\eqref{lp:multi7} only.
	Thus, we only need to show that, given a receiver $r$ and an arbitrary menu of marginal signaling schemes $\{\phi^{r,k}\}_{k \in \K_r}$, we can always build a menu of \emph{direct} marginal signaling scheme $\{\bar \phi^{r,k}\}_{k \in \K_r}$ such that $\bar \phi^{r,k}_\theta(a_1)\ge \phi^{r,k}_\theta(a_1)$  for each $\theta \in \Theta$ and $k \in \K_r$.  By the monotonicity assumption on $f_\theta$ the optimal sender's strategy with marginal signaling scheme $\{\bar \phi^{r,k}\}_{r \in \rec,k \in \K_r}$ has an utility greater or equal to the one with $\{\phi^{r,k}\}_{r \in \rec, k \in \K_r}$.

	This can be proved by following steps similar to those of Lemma~\ref{lm:finite} and Theorem~\ref{thm:single} for each menu of marginal signaling schemes.
	%
	%
	In particular,
	let $r$ be a receiver and $\Phi^r=\{\phi^{r,k}\}_{k \in \K_r}$ be a menu of marginal signaling schemes that induces probability distribution $\gamma^{r,k}$ over the posteriors when the reported type is $k$.
	Notice that the probability that the receiver of type $k \in \K_r$ plays an action $a \in \A_r$ is given by $\sum_{\xi \in \hat \Xi^{k,a_1}}  \gamma^{r,k}_\pvec \xi_\theta$.\footnote{For the ease of presentation, we assume that $\gamma^{r,k}$ has finite support. Formally, we should replace $\sum_{\xi \in \hat \Xi^{k,a}}  \gamma^{r,k}_\pvec \xi_\theta$ with  $Pr_{\pvec \sim \gamma^{r,k}}\left\{\pvec \in \hat \Xi^{k,a} \right\} \mathbb{E} \left[ \xi_\theta \mid \pvec \in \hat \Xi^{k,a} \right]$.} 
	 We can obtain a menu of probability distribution $\{\bar \gamma^{r,k}\}_{k \in \K_r}$ over $\Xi^*$ such that the probability that the receiver plays $a_1$ increases for each $k$ and $\theta$, \emph{i.e.}, $\sum_{\xi \in \hat \Xi^{k,a_1}\cap \Xi^*} \bar \gamma^{r,k}_\pvec \xi_\theta \ge \sum_{\xi \in \hat \Xi^{k,a_1}}  \gamma^{r,k}_\pvec \xi_\theta$. To see this, it is sufficient to follow the proof of Lemma \ref{lm:finite} and notice that the receiver always breaks ties in favor of $a_1$ by the monotonicity assumption on $f_\theta$.
	Finally, setting $\bar \phi^{r,k}_\theta(a)=\frac{\sum_{\pvec \in \hat \pset^{k,a} \cap \pset^*} \bar \gamma^k_\pvec \p_\theta}{\mu_\theta}$ for each $k \in \K_r$, $\theta \in \Theta$ and $a \in \A_r$, we obtain a menu of signaling schemes such that
	\[\bar \phi^{r,k}_\theta(a_1)= \frac{\sum_{\pvec \in \hat \pset^{k,a_1} \cap \pset^*} \bar \gamma^k_\pvec \p_\theta}{\mu_\theta} \ge \frac{  \sum_{\xi \in \hat \Xi^{k,a_1}}  \gamma^{r,k}_\pvec \xi_\theta}{\mu_\theta}  \quad \forall \theta \in \Theta, \forall k \in \K_r\]
	
	To conclude, following the proof of Theorem \ref{thm:single} we can show that the menu of marginal signaling schemes $\{\bar \phi^{r,k}\}_{k \in \K_r}$ is IC and persuasive.
\end{proof}

\theoremMulti*

\begin{proof}
	Since \ref{lp:multi} has an exponential number of constraints, we work on the dual formulation.

	\begin{subequations}
		\begin{align}
		\min_{q,t\le0,z\le0,y\le0,p} & - \sum_{r \in \rec,k\neq k' \in \K_r} \sum_\theta\mu_\theta u^{r,k}_\theta(a_0) t_{r,k,k'} + \sum_{r \in \rec,a \in \A_r,k\neq k' \in \K_r} \sum_{\theta \in \Theta} \mu_\theta u^{r,k}_\theta(a) z_{a_0,r,a,k,k'} \\
		&\hspace{6.5cm}- \sum_{r \in \rec,k \in \K_r} \sum_\theta \mu_\theta [u^{r,k}_\theta(a_0)-u^{r,k}_\theta(a_1)] y_{a_0,r,k}+\sum_{\kvec \in \bar \K,\theta \in \Theta} p_{\kvec,\theta}\\
		&\hspace{-1.5cm}- \sum_{\kvec' \in \bar \K:k'_r=k}q_{\kvec',r,\theta} +\sum_{k'\neq k} t_{r,k,k'} [ \mu_\theta u^{r,k}_\theta(a_1) - \mu_\theta u^{r,k}_\theta(a_0) ]-   \sum_{a \in \A_r,k' \neq k} \mu_\theta u^{k'}_\theta(a) z_{a_1,r,a',k',k} +  \sum_{a \in \A_r,k' \neq k} \mu_\theta u^{k'}_\theta(a) z_{a_0,r,a',k',k}\\
		&\hspace{2.3cm}+\mu_\theta[u^{r,k}_\theta(a_1)-u^{r,k}_\theta(a_0)]y_{a_1,r,k}-\mu_\theta[u^{r,k}_\theta(a_0)-u^{r,k}_\theta(a_0)]y_{a_0,r,k} \ge 0 \hspace{0.5cm} \forall r \in \rec ,k \in \K_r, \forall \theta \in \Theta \label{lp:multidual1}\\
		& \hspace{-1.1cm} - t_{r,k,k'}+ \sum_{a' \in \A_r}z_{a_1,r,a',k,k'} \ge 0 \hspace{9cm}\forall r \in \rec,k\neq k' \in \K_r\\
		& \hspace{-1.1cm} - t_{r,k,k'}+ \sum_{a' \in \A_r}z_{a_0,r,a',k,k'} \ge 0 \hspace{9cm} \forall r \in \rec,k\neq k' \in \K_r \\
		& \hspace{-1.1cm}\sum_{ r \in R} q_{\kvec,r,\theta} + p_{\kvec,\theta} \ge \mu_\theta \lambda_\kvec  f_\theta(R) \hspace{8.5cm}\forall \kvec \in \bar \K, \forall \theta \in \Theta, \forall R \subseteq \rec \label{lp:multidual2} 
		\end{align}
	\end{subequations}
	
	where variables $q_{\kvec,r,\theta}$ relative to constraints \eqref{lp:multi1}, $t_{r,k,k'}$ to \eqref{lp:multi2}, $z_{a,r,a',k,k'}$ to \eqref{lp:multi3} and \eqref{lp:multi4}, $y_{a,r,k}$ to \eqref{lp:multi5} and \eqref{lp:multi6}, $p_{\kvec,\theta}$ to \eqref{lp:multi7}.
	
	%
	To solve the problem with the ellipsoid method it is sufficient to design a polynomial time separation oracle.
	We focus on the separation oracle that returns a violated constraint.
	Given an assignment to the variables, there are a polynomial number of constraint \ref{lp:multidual1} (with polynomially many variables) and we can check if a constraint is violated in polynomial time.
	Moreover, for each $\bar \theta$, $\bar \kvec$, we can find if there exists a violated constraint $(\bar \kvec,\bar \theta,R)$. We can use the oracle to find $max_{R \subseteq \rec}  \lambda_k f_\theta(R)- \sum_{r:r \in R} q_{\bar\kvec,r,\bar\theta} $. If it is greater than $p_{\bar \kvec,\bar \theta}$, we can return a violated constraint, while if it is smaller or equal to $p_{\bar \kvec, \bar \theta}$, all the constraints $\{(\bar \kvec,\bar \theta,R)\}_{R \subseteq \rec}$ are satisfied.
\end{proof}

\corollaryOracle*
\begin{proof}
	By Theorem \ref{thm:oracle}, we only need to design a polynomial time oracle. Since the sum of a supermodular and a modular function is supermodular, and unconstrained supermodular maximization can be solved in polynomial time, an oracle can be designed in polynomial time for supermodular functions.
	For anonymous functions we can construct a polynomial time oracle as follows. We can enumerate over all $n\in \{ 0,\dots,|\rec|\}$. Once we fix the size of the set to $n$, the optimal set includes the $n$ receiver with higher values of weights $w$.
\end{proof}

\section{Proofs Omitted from Section~\ref{sec:results_multiSub}}

\lemmaSubUniform*

\begin{proof}
	Given an optimal solution $(\phi,x)$ to LP~\ref{lp:multi}, we show how to build a solution to LP~\ref{lp:multiSub} with almost the same value.
	Since LP \eqref{lp:multi} has $\beta$ constraints, there exists an optimal solution $(\phi,x)$ to LP \ref{lp:multi} with support at most $\beta$.
	We construct a solution to Program~\ref{lp:multiSub} with the same values of variables $x_\theta^{r,k}$ (representing marginalsignalign schemes).
	Then, we show how to obtain a $q$-uniform distribution for every $\kvec \in \K$ and $\theta \in \Theta$.
	Fix $\kvec \in K$ and $\theta\in \Theta$.
	Let $G^{\kvec,\theta} \subseteq 2^\rec$ be the subsets of $R \subseteq \rec$ that are in the support of distribution $\phi^{\kvec}_\theta$, namely $\phi^{\kvec}_\theta(R)>0$.
	Notice that $| G^{\kvec,\theta} | \le \beta$, since the solution has support at most $\beta$.
	For every $R \in G^{\kvec,\theta}$, we define $N^{\kvec,\theta}(R)$ as the greatest integer $i$ such that $\phi^{\kvec}_\theta(R) \ge \frac{i}{q}$.
	Finally, for every $R\in G^{\kvec,\theta}$, we choose  $N^{\kvec,\theta}(R)$ indexes $j \in [q]$ (with each index being selected at most one time) for which we set $x^{j,\kvec,\theta}_r=1$ for every $r \in R$, and $x^{j,\kvec,\theta}_r=0$ for every $r \notin R$.
	%
	Since $\sum_{R \in G^{\kvec,\theta}} |N^{\kvec,\theta}(R)| \le  \sum_{R \in G^{\kvec,\theta}} q\phi^{\kvec}_\theta(R)=q$, we have defined values for at most $q$ indexes.
	For all the remaining indexes $j \in [q]$, we set $x^{j,\kvec,\theta}_r=0$ for $r \in \rec$.

	It is easy to see that the defined solution is feasible since, for every $\kvec \in \K$, $\theta \in \Theta$, and $r \in \rec$, it holds that
	\[\sum_{j\in [q]} \frac{1}{q} x^{j,\kvec,\theta}_{r} = \frac{1}{q} \sum_{R \in G^{\kvec,\theta}:r \in R} \hspace{-2mm}N^{\kvec,\theta}(R)   \le \sum_{R \in G^{\kvec,\theta}:r \in R} \hspace{-2mm} \phi^{\kvec}_\theta(R)= x^{r,k_r}_\theta.\]
	%
	Moreover, for every $\kvec \in \K$ and $\theta \in \Theta$, the sender's expected utility in a state of nature $\theta \in \Theta$ is at least  
	\begin{align*} 
	\frac{1}{q} \sum_{j \in [q]} F_\theta \left( x^{j,\kvec,\theta} \right) &= \frac{1}{q}\sum_{R \in G^{\kvec,\theta}} N^{\kvec,\theta}(R) f_\theta(R) \\
	&\ge \sum_{R \in G^{\kvec,\theta}} \left( \phi^\kvec_\theta(R)f_\theta(R)-\frac{1}{q} \right) \\
	&\ge\sum_{R \in G^{\kvec,\theta}} \phi^\kvec_\theta(R) f_\theta(R)- \frac{\beta}{q} \\
	&\ge \sum_{R \subseteq \rec} \phi^\kvec_\theta(R) f_\theta(R)- \epsilon,
	\end{align*}
	where the equality follows from $x^{j,\kvec,\theta}_r\in \{0,1\}$, the first inequality by $\frac{1}{q}N^{\kvec,\theta}(R) \ge \phi^\kvec_\theta(R)-\frac{1}{q}$, the second one from the fact that $|G^{\kvec,\theta}| \le \beta$, and the last one by the definitions of $q$ and $G^{\kvec,\theta}$.
	Hence, the sender's expected utility is at least 
	\begin{align*}
	\sum_{\theta \in \Theta} \mu_\theta \sum_{\kvec \in \K} \lambda_\kvec \frac{1}{q} & \sum_{j \in [q]}  F_\theta( x^{j,\kvec,\theta}) \\
	&\ge \sum_{\theta \in \Theta} \mu_\theta \sum_{\kvec \in \K} \lambda_\kvec \left(\sum_{R \subseteq \rec} \phi^k_\theta(R) f_\theta(R)-\epsilon\right)\\
	&=\sum_{\theta \in \Theta} \mu_\theta \sum_{\kvec \in \K} \lambda_\kvec \sum_{R \subseteq \rec} \phi^k_\theta(R) f_\theta(R)-\epsilon
	\end{align*}
	This concludes the proof.
\end{proof}

\end{document}